\documentclass[graybox]{svmult}

\usepackage{type1cm}
\usepackage{graphicx}
\usepackage{newtxtext}
\usepackage[varvw]{newtxmath}
\usepackage{overpic}
\graphicspath{{Figures/}}
\usepackage{adjustbox}
\usepackage{subfigure}
\usepackage{url}
\usepackage{listings}
\usepackage[most]{tcolorbox}
\usepackage{nicematrix}

\usepackage{xcolor}
\definecolor{myblue}{HTML}{0086D0}
\definecolor{Emerald}{rgb}{0.31, 0.78, 0.47}
\definecolor{MidnightBlue}{rgb}{0.1, 0.1, 0.44}
\definecolor{YellowOrange}{rgb}{0.95, 0.52, 0.0}
\definecolor{Magenta}{rgb}{1.0, 0.0, 1.0}
\definecolor{Salmon}{rgb}{1.0, 0.57, 0.64}
\definecolor{LimeGreen}{rgb}{0.462745098, 0.725490196, 0.000000000} %

\definecolor{codegreen}{rgb}{0,0.6,0}
\definecolor{codegray}{rgb}{0.5,0.5,0.5}
\definecolor{codepurple}{rgb}{0.58,0,0.82}
\definecolor{backcolour}{rgb}{0.95,0.95,0.92}

\lstdefinestyle{mystyle}{
    commentstyle=\color{codegreen},
    keywordstyle=\color{magenta},
    numberstyle=\tiny\color{codegray},
    stringstyle=\color{codepurple},
    basicstyle=\ttfamily\fontsize{4.5}{4}\selectfont,
    captionpos=b,                    
    keepspaces=true,                 
    showspaces=false,                
    showstringspaces=false,
    showtabs=false,                  
    tabsize=2
}
\lstdefinestyle{mystyle3}{
    commentstyle=\color{codegreen},
    keywordstyle=\color{magenta},
    numberstyle=\tiny\color{codegray},
    stringstyle=\color{codepurple},
    basicstyle=\ttfamily\fontsize{3.5}{3}\selectfont,
    captionpos=b,                    
    keepspaces=true,                 
    showspaces=false,                
    showstringspaces=false,
    showtabs=false,                  
    tabsize=2
}

\lstdefinestyle{mystyle2}{
    commentstyle=\color{codegreen},
    keywordstyle=\color{magenta},
    numberstyle=\tiny\color{codegray},
    stringstyle=\color{codepurple},
    basicstyle=\ttfamily\fontsize{6.5}{6}\selectfont,
    captionpos=b,                    
    keepspaces=true,                 
    showspaces=false,                
    showstringspaces=false,
    showtabs=false,                  
    tabsize=2
}

\begin{document}
\title*{Generator Matrices by Solving Integer Linear Programs}
\author{Lo\"is Paulin, David Coeurjolly, Nicolas Bonneel, Jean-Claude Iehl, Victor Ostromoukhov, and Alexander Keller}
\authorrunning{Paulin et al.}
\institute{Lo\"is Paulin \email{lois.paulin@liris.cnrs.fr}
\and David Coeurjolly  \email{david.coeurjolly@cnrs.fr}
\and Nicolas Bonneel  \email{nicolas.bonneel@liris.cnrs.fr}
\and Jean-Claude Iehl \email{jean-claude.iehl@liris.cnrs.fr}
\and Victor Ostromoukhov \email{victor.ostromoukhov@liris.cnrs.fr} \at Universit\'e de Lyon, UCBL, CNRS, INSA Lyon, LIRIS
\and Alexander Keller \email{akeller@nvidia.com} \at NVIDIA, Fasanenstr. 81, 10623 Berlin, Germany}

\maketitle

\abstract*{In quasi-Monte Carlo methods, generating high-dimensional low discrepancy sequences by generator matrices is a popular and efficient approach.
Historically, constructing or finding such generator matrices has been a hard problem.
In particular, it is challenging to take advantage of the intrinsic structure of a given numerical problem to design samplers of low discrepancy in certain subsets of dimensions.
To address this issue, we devise a greedy algorithm allowing us to translate desired net properties into linear constraints on the generator matrix entries.
Solving the resulting integer linear program yields generator matrices that satisfy the desired net properties.
We demonstrate that our method finds generator matrices in challenging settings, offering low discrepancy sequences beyond the limitations of classic constructions.}

\abstract{In quasi-Monte Carlo methods, generating high-dimensional low discrepancy sequences by generator matrices is a popular and efficient approach.
Historically, constructing or finding such generator matrices has been a hard problem.
In particular, it is challenging to take advantage of the intrinsic structure of a given numerical problem to design samplers of low discrepancy in certain subsets of dimensions.
To address this issue, we devise a greedy algorithm allowing us to translate desired net properties into linear constraints on the generator matrix entries.
Solving the resulting integer linear program yields generator matrices that satisfy the desired net properties.
We demonstrate that our method finds generator matrices in challenging settings, offering low discrepancy sequences beyond the limitations of classic constructions.}


\section{Introduction}

Monte Carlo and quasi-Monte Carlo integration are an important part of many numerical algorithms, for example, in computational finance and image synthesis.  
Problems in these domains often have an intrinsic structure that allows one to benefit from a particular uniformity profile imposed on selected subsets of dimensions. This uniformity may be specified by $(t,m,s)$-net properties, i.e. low discrepancy properties. A popular way to compute low discrepancy sequences is via generator matrices. 
Finding generator matrices that match a specific uniformity profile has been a challenging problem.
Low discrepancy sequences, such as Sobol’, Halton, and others, use specific matrix constructions to achieve provably strong properties.
However, these constructions represent only a tiny fraction of all possible generator matrices and hence may not be able to ideally match a desired uniformity profile.
As an example, Joe and Kuo's work~\cite{joe2008constructing} on optimizing the uniformity of all pairs of
dimensions demonstrates that the resulting generator matrices do not reach a satisfactory 2D
uniformity.
Methods for optimizing general matrices have been devised~\cite{sun2002digital}. They often rely on the random generation and selection of generator matrices.
This approach fails when the desired uniformity profiles are too restrictive.
Furthermore, some sets of uniformity profiles are provably infeasible in small prime bases such as $b=2$.

We propose a linear algebra-based understanding of the fundamental theorems of digital nets and their generator matrices.
We use this understanding to convert projective uniformity properties into polynomial and linear constraints on generator matrix entries.
Numerically solving for these constraints allows us to construct generator matrices with net structures that have not been achieved before.
Our method works in higher prime bases, too, offering one more degree of freedom to satisfy a specific set of projective net properties.

\section{Quasi-Monte Carlo Methods}

Quasi-Monte Carlo~\cite{Niederreiter1992} integration
estimates the value of the integral of $f$ by 
\[
  \int_{[0,1)^s} f(\vec u) d \vec u \approx \frac{1}{N} \sum_{i=0}^{N-1} f(\vec x_i) \, .
\]
Rather than using independent realizations of uniformly, identically distributed random variables,
low discrepancy sequences are used to generate deterministic samples  $\vec x_i$ that
are much more uniformly distributed over the $s$-dimensional unit cube $[0,1)^s$ than uniform random samples can be.

The uniformity of a point set $P = \{ \vec x_0, \ldots , \vec x_{N - 1}\}$ can be measured
by the star-discrepancy~\cite[p.~14]{Niederreiter1992}
\[
  D_N^*(P) := \sup_{B} \left| \lambda_s(B) - \frac{1}{N} \sum_{i = 0}^{N - 1} c_B(\vec x_i) \right| \mbox{ ,}
\]
where the supremum is taken over all sets $B$ of the form
$B=\prod_{i=1}^s[0,b_i]$, with $0 \leq b_i \leq 1$,
$\lambda_s$ is the Lebesgue measure of $B$, and
$c_B$ is the characteristic function of the set $B$ that is one for $\vec x_i \in B$ and zero otherwise.
The star-discrepancy can be interpreted as the worst-case integration
error of the class of characteristic functions $c_B$. The more uniformly the points
are distributed, the smaller the star-discrepancy is. Informally, a point set
is called low discrepancy if its star-discrepancy is small. For
low-discrepancy sequences of points, the star-discrepancy vanishes in
${\mathcal O}\left( \frac{\log^s N}{N} \right)$, while random points can only
achieve a discrepancy of order ${\mathcal O}\left( \sqrt{\frac{\log \log N}{N}} \right)$.

For functions of bounded variation $V$ in the sense of Hardy and Krause,
the Koksma-Hlawka inequality
\[
  \left| \frac{1}{N} \sum_{i=0}^{N-1} f(\vec x_i) - \int_{[0,1)^s} f(\vec u) d \vec u \right| \leq V(f) \cdot D_N^*(P) \, ,
\]
bounds the integration error by the product of variation and star-discrepancy.
Even though in practice many functions may not be of bounded
variation, numerical experiments suggest that increasing the uniformity
of the points by decreasing the discrepancy may lower the integration error.

\subsection{Digital $(t,s)$-Sequences and $(t, m, s)$-Nets}

Deterministic digital $(t,s)$-sequences and $(t, m, s)$-nets~\cite{Niederreiter1992}
are low discrepancy point sequences and sets that can be efficiently generated.
Given generator matrices $C_1, \ldots, C_s$, the algorithm
\begin{equation} \label{Eqn:Algorithm}
   x_i^{(j)} = \left(\begin{array}{c}b^{-1} \\ \vdots \\ b^{-m}\end{array}\right)^T
        \cdot \underbrace{C_j \cdot \left(\begin{array}{c}i_1(i) \\ \vdots \\ i_{m}(i) \end{array}\right)}_\text{multiplication in $\mathbb{F}_b$}
\end{equation}
first multiplies each generator matrix $C_j$ by the point index $i = \sum_{k = 1}^\infty i_k(i) b^{k-1}$
represented as vector of $m$ digits $i_k$ in the integer base $b$. The result of the matrix multiplication
in a finite field $\mathbb{F}_b$ is then mapped to the unit interval $[0,1)$ by a scalar product
with the vector of inverse powers of the base $b$. While we work with prime bases, using a prime power requires mapping the digits
to the finite field and mapping the resulting digits back into the integers~\cite[Sec.~4.3]{Niederreiter1992}.

The stratification and hence uniformity
properties of this digital construction are characterized by elementary intervals as given by

\begin{figure}
    \centering
    \setlength{\unitlength}{0.2\linewidth}
    \begin{tabular}{ccccc}
    \begin{picture}(1,1)(0,0)
    \thicklines
    \color{orange}
    \multiput(0, 0.25)(0, 0.25){3}{\line(1, 0){1}}
    \color{myblue}
    \multiput(0.5, 0)(0.5,0){1}{\line(0,1){1}}
    \color{black}
    \put(0,0){\line(1,0){1}}
    \put(0,0){\line(0,1){1}}
    \put(1,0){\line(0,1){1}}
    \put(0,1){\line(1,0){1}}
    \end{picture}
    & & \begin{picture}(1,1)(0,0)
    \thicklines
    \color{orange}
    \multiput(0, 0.25)(0, 0.25){3}{\line(1, 0){1}}
    \color{myblue}
    \multiput(0.25, 0)(0.25,0){3}{\line(0,1){1}}
    \color{black}
    \put(0,0){\line(1,0){1}}
    \put(0,0){\line(0,1){1}}
    \put(1,0){\line(0,1){1}}
    \put(0,1){\line(1,0){1}}
    \end{picture}
    & & \begin{picture}(1,1)(0,0)
    \thicklines
    \color{orange}
    \multiput(0, 0.333333)(0, 0.33333){2}{\line(1, 0){1}}
    \color{myblue}
    \multiput(0.111111, 0)(0.111111,0){8}{\line(0,1){1}}
    \color{black}
    \put(0,0){\line(1,0){1}}
    \put(0,0){\line(0,1){1}}
    \put(1,0){\line(0,1){1}}
    \put(0,1){\line(1,0){1}}
    \end{picture} \\
    $b=2$, $\vec d=(1,2)$
    & & $b=2$, $\vec d=(2,2)$
    & & $b=3$, $\vec d=(2,1)$
    \end{tabular}
    \caption{Examples of elementary intervals in bases $b = 2$ and $b = 3$ in $s = 2$ dimensions.
    The vectors $\vec d = (d_1, d_2)$ determine the resolutions $b^{-d_j}$ along
    the canonical axes, resulting in the depicted stratification.}
    \label{Fig:elem_intervals}
\end{figure}
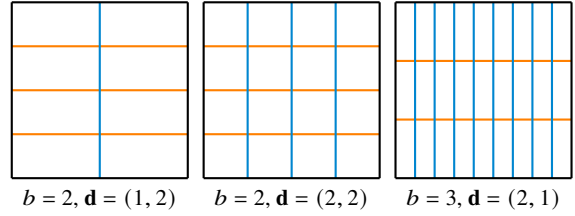

\begin{definition}[see~{\cite[p.~48]{Niederreiter1992}}] \label{Def:Elementary}
      An interval of the form
      \[
        E := \prod_{j=1}^s \left[a_j b^{-d_j}, (a_j + 1) b^{-d_j}\right) \subseteq [0,1)^s\,,
      \]
      for $0 \leq a_j < b^{d_j}$ and integers $d_j \geq 0$
      is called an \emph{elementary interval in base~$b$}.
\end{definition}

The elementary intervals (see Fig.~\ref{Fig:elem_intervals}) are used
to characterize the stratification of point sets as specified in

\begin{definition}[see {\cite[Def.~4.1]{Niederreiter1992}}] \label{Def:tmsNet}
      For integers $0 \leq t \leq m$, a \emph{$(t,m,s)$-net in base~$b$} is a point set
      of $b^m$ points in $[0,1)^s$ such that there are exactly $b^t$ points in each elementary
      interval $E$ with volume $b^{t-m}$.
\end{definition}

\begin{figure}
        \centering
          \includegraphics[trim={5.25cm 8cm 5.25cm 8cm}, clip, width=0.9\textwidth]{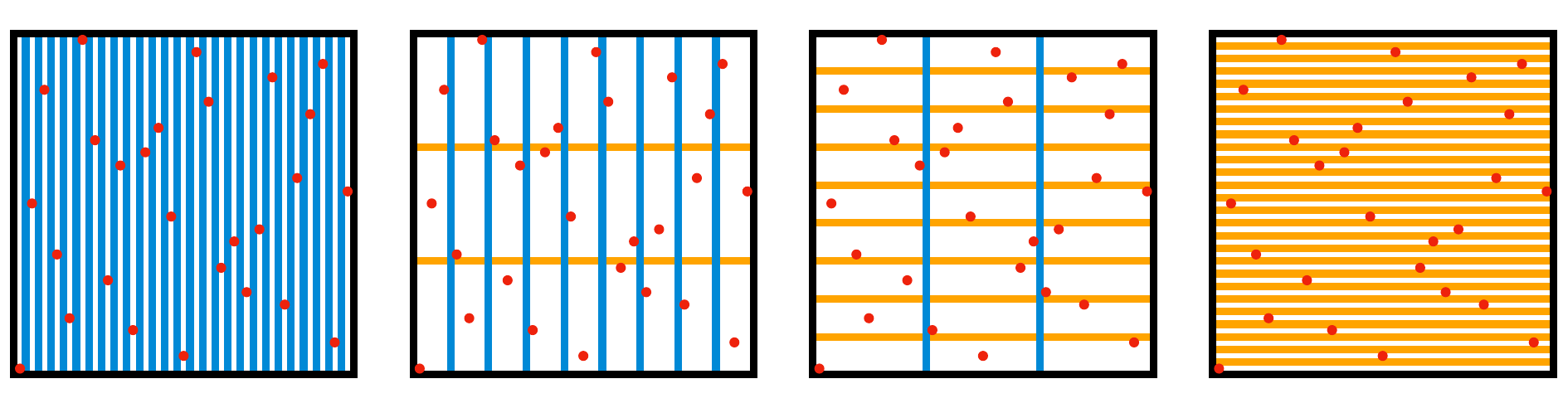}
      \caption{Example of a $(0,3,2)$-net in base $b = 3$. There are $b^m = 3^3 = 27$ points in $s = 2$ dimensions and each elementary interval (subdividing the $x$ and $y$ coordinates into $b^i = 3^i$, $i=1..m$, intervals) contains exactly $b^t = 3^0 = 1$ point.}
      \label{Fig:tmsnet}
\end{figure}  

Fig.~\ref{Fig:tmsnet} illustrates the structure of the elementary
intervals for the example of a $(0,3,2)$-net in base $b = 3$. Note that $\sum_{j=1}^s d_j = m - t$
by the definitions of the elementary intervals and the $(t, m, s)$-nets.
The structure of $(t, m, s)$-nets can be extended to sequences of points: 

\begin{definition}[see {\cite[Def.~4.2]{Niederreiter1992}}] \label{Def:tsSequence}
      For an integer $t \geq 0$, a sequence $\vec{x}_0,\vec{x}_1,\dots$ of points in $[0,1)^s$
      is a \emph{$(t,s)$-sequence in base~$b$} if, for all integers $k \geq 0$ and $m > t$,
      the point set $\vec{x}_{kb^m},\dots,\vec{x}_{(k+1)b^m-1}$ is a $(t,m,s)$-net in base~$b$.
\end{definition}

\begin{figure}
    \[
    \begin{array}{ccc}
    \begin{pNiceMatrix}[margin]
    \Block[fill=myblue!35]{1-3}{}
    c^{(1)}_{1,1} & c^{(1)}_{1,2} & c^{(1)}_{1,3} \\
    c^{(1)}_{2,1} & c^{(1)}_{2,2} & c^{(1)}_{2,3} \\
    c^{(1)}_{3,1} & c^{(1)}_{3,2} & c^{(1)}_{3,3} \\
    \end{pNiceMatrix}
    \cdot
    \begin{pNiceMatrix}
    i_1 \\ i_2 \\ i_3
    \end{pNiceMatrix}
    & = & 
    \begin{pNiceMatrix}[margin]
     \Block[fill=myblue!35]{1-1}{}
    e^{(1)}_1 \\ e^{(1)}_2 \\ e^{(1)}_3
    \end{pNiceMatrix} \\ \\
    \begin{pNiceMatrix}[margin]
    \Block[fill=orange!35]{2-3}{}
    c^{(2)}_{1,1} & c^{(2)}_{1,2} & c^{(2)}_{1,3} \\
    c^{(2)}_{2,1} & c^{(2)}_{2,2} & c^{(2)}_{2,3} \\
    c^{(2)}_{3,1} & c^{(2)}_{3,2} & c^{(2)}_{3,3} \\
    \end{pNiceMatrix}
    \cdot
    \begin{pNiceMatrix}
    i_1 \\ i_2 \\ i_3
    \end{pNiceMatrix}
    & = & 
    \begin{pNiceMatrix}[margin]
     \Block[fill=orange!35]{2-1}{}
    e^{(2)}_1 \\ e^{(2)}_2 \\ e^{(2)}_3
    \end{pNiceMatrix}
    \end{array}
    \qquad
    \raisebox{-1cm}{\setlength{\unitlength}{0.18\linewidth}
    \begin{picture}(1,1)(0,0)
    \thicklines
    \color{orange}
    \multiput(0, 0.25)(0, 0.25){3}{\line(1, 0){1}}
    \color{myblue}
    \multiput(0.5, 0)(0.5,0){1}{\line(0,1){1}}
    \color{black}
    \put(0,0){\line(1,0){1}}
    \put(0,0){\line(0,1){1}}
    \put(1,0){\line(0,1){1}}
    \put(0,1){\line(1,0){1}}
    \put(0.15,0.08){\color{myblue}0\color{orange}00}
    \put(0.65,0.08){\color{myblue}1\color{orange}00}
    \put(0.15,0.33){\color{myblue}0\color{orange}01}
    \put(0.65,0.33){\color{myblue}1\color{orange}01}
    \put(0.15,0.58){\color{myblue}0\color{orange}10}
    \put(0.65,0.58){\color{myblue}1\color{orange}10}
    \put(0.15,0.83){\color{myblue}0\color{orange}11}
    \put(0.65,0.83){\color{myblue}1\color{orange}11}
    \end{picture}}
    \qquad
     \underbrace{\begin{pNiceMatrix}[margin]
    \Block[fill=myblue!35]{1-3}{}
    c^{(1)}_{1,1} & c^{(1)}_{1,2} & c^{(1)}_{1,3} \\
    \Block[fill=orange!35]{2-3}{}
    c^{(2)}_{1,1} & c^{(2)}_{1,2} & c^{(2)}_{1,3} \\
    c^{(2)}_{2,1} & c^{(2)}_{2,2} & c^{(2)}_{2,3} \\
    \end{pNiceMatrix}}_{= M_{\vec k}}
    \cdot
    \begin{pNiceMatrix}
    i_1 \\ i_2 \\ i_3
    \end{pNiceMatrix}
    = 
    \begin{pNiceMatrix}[margin]
    \Block[fill=myblue!35]{1-1}{} e^{(1)}_1 \\
    \Block[fill=orange!35]{2-1}{}
    e^{(2)}_1 \\
    e^{(2)}_2
    \end{pNiceMatrix}
    \]

    \caption{Illustration of the algebraic relationship of generator matrices and elementary intervals.
    Left: As highlighted, the first few rows of a generator matrix determine the most significant digits $e_k$.
    Middle: For the displayed set of elementary intervals, the digit $e^{(1)}_1$ determines the partition along
    the first dimension, while $e^{(2)}_1$ and $e^{(2)}_2$ select the partition along the second dimension.
    Right: If now $\text{det}(M_{\vec k}) \not = 0$, the composite matrix $M_{\vec k}$ will be a bijection
    between the elementary intervals and the index $(i_1, i_2, i_3)$.}
    \label{Fig:GeneratorMatrices}
\end{figure} 

\subsection{Relating Generator Matrices and Elementary Intervals}
\label{Sec:Constraints}

Fig.~\ref{Fig:GeneratorMatrices} illustrates the relationship between generator
matrices and elementary intervals: multiplying a generator matrix $C_j$
by an index vector (see Equation~\eqref{Eqn:Algorithm}), the $k$ most significant digits of the result are
determined by the first $k$ rows of the generator matrix. Then, a matrix
$M_{\mathbf{k}}$ composed of the first $k_j$ rows of $C_j$,
where $\mathbf{k}:=\left( k_1, \dots, k_s \right) \in \mathbb{N}_0^s$ and $\sum_{j=1}^s k_j = m$,
defines a mapping from the indices to elementary intervals of
size $\left( \frac{1}{b^{k_1}}, \dots, \frac{1}{b^{k_s}} \right)$ and consequently of volume $\frac{1}{b^m}$.
If and only if $\text{det}(M_{\vec k}) \not = 0$ in the Galois field $\mathbb{F}_b$, the mapping will be bijective.
Hence the determinants of the matrices $M_{\vec k}$ can be used
to verify the properties of $(0,m,s)$-nets, where each elementary interval
contains exactly one point.

For $t > 0$, we select $\mathbf{k} := \left( k_1,  \dots, k_s\right)$ such that $\sum_{j=1}^s k_j = m - t$.
This results in rectangular matrices $M_{\mathbf{k}}$ defining a mapping between indices and elementary intervals of volume $\frac{1}{b^{m-t}}$.
We want each elementary interval to be the image of exactly $b^t$ indices, i.e. contain exactly $b^t$ points.
This is equivalent to having $\dim(\text{ker } M_{\mathbf{k}}) = t$ and $\dim(\text{im } M_{\mathbf{k}}) = m - t$, because the dimension of the index vector in Equation~\eqref{Eqn:Algorithm} is $m$.
$M_{\mathbf{k}}$ having $m-t$ rows, this means $M_{\mathbf{k}}$ needs
to have full rank in the Galois field $\mathbb{F}_b$.
As such, the generator matrices $C_1, \ldots, C_s$ define a $(t,m,s)$-net if and only if $\forall \mathbf{k}=\left(k_1,\dots,k_s\right)$ with $\sum_{j=1}^s k_j = m - t$, $M_{\mathbf{k}}$ is of full rank.
The same illustration as in Fig.~\ref{Fig:GeneratorMatrices} applies, except that now the $t$ least significant
digits of the index $i$ are not considered.

\begin{figure}
\centering
\begin{overpic}[width=\linewidth]{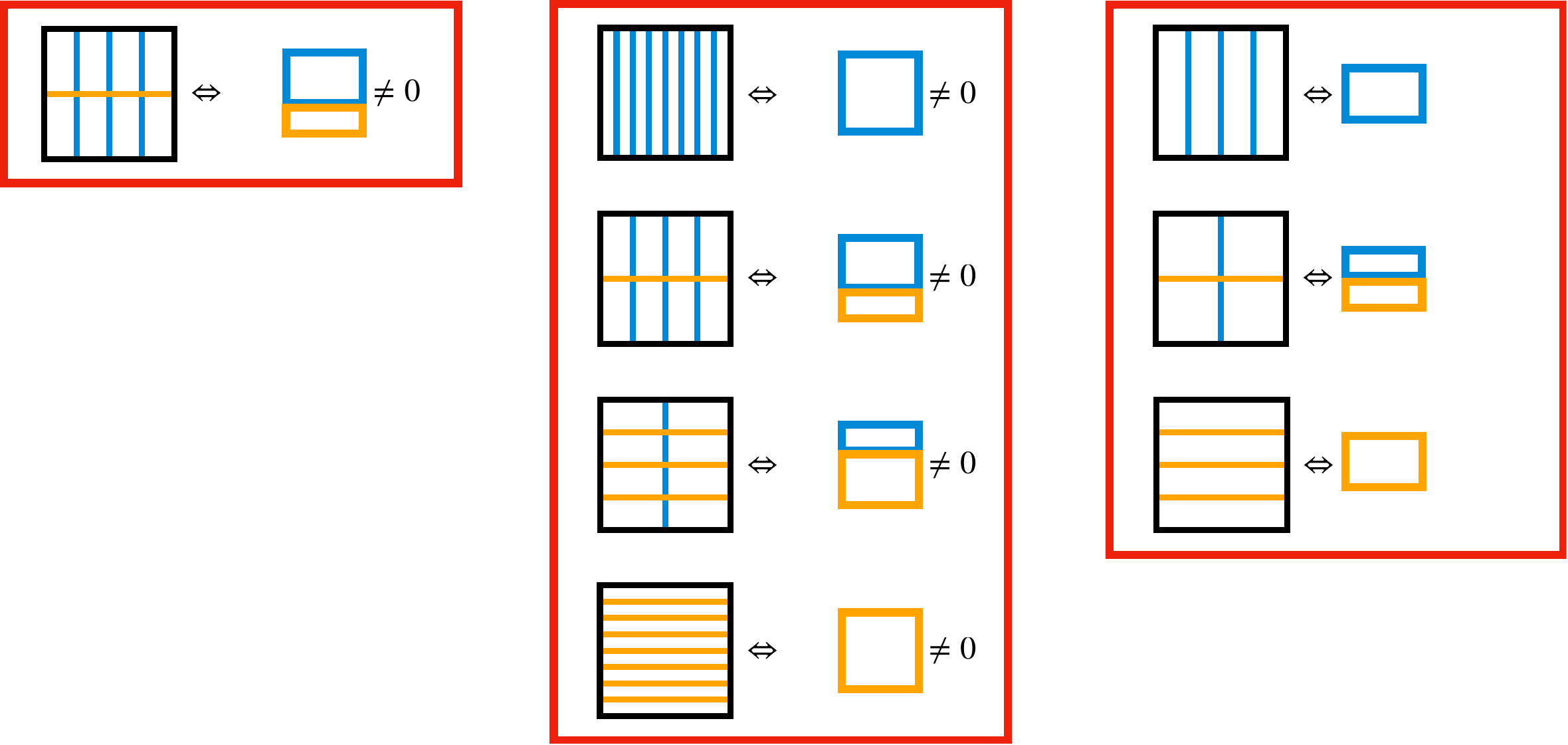}\small
  \put(10,48){Stratified}
  \put(43,48){$(0,m,s)-$net}
  \put(80,48){$(t,m,s)-$net}
  \put(14.5,41){det}
  \put(50,41){det}
  \put(50,29.3){det}
  \put(50,17.4){det}
  \put(50,5.5){det}
  \put(92,41){\scriptsize full rank}
  \put(92,29.3){\scriptsize full rank}
  \put(92,17.4){\scriptsize full rank}
\end{overpic}
\caption{Stratification properties (left) expressed by matrix determinants
allow for checking $(0,m,s)$-net properties (middle) and $(t,m,s)$-net properties
for $t > 0$ (right). The red boxes group sets of constraints as explained in Sec.~\ref{Sec:Constraints}.} %
\label{Fig:Overview}
\end{figure}
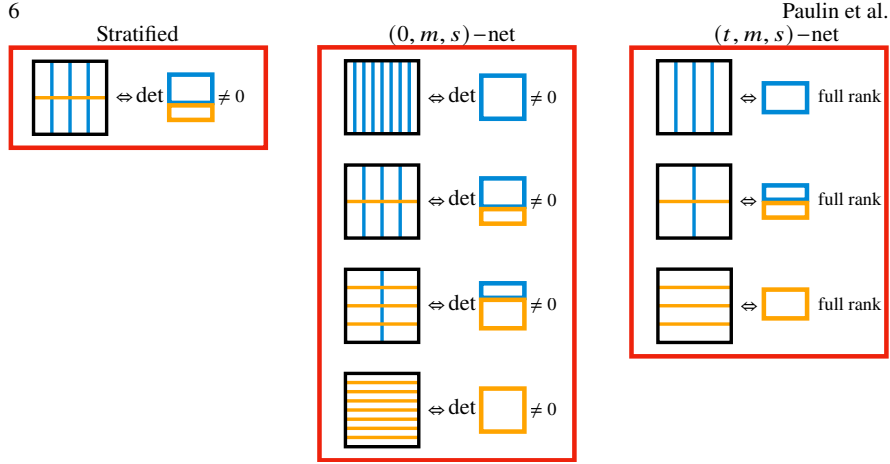

With parameters $t \geq 0$ at hand, some techniques to relax constraints with $t=0$
within the prior MatBuilder~\cite{matbuilder2022} 
become dispensable as demonstrated by the results in Sec.~\ref{Sec:Results}. A first
example is generalized stratification, where for $b^m$ points and
\[
           \forall \left(k_1, \dots, k_{s}\right) 
  \in \left\{ \left\lfloor\frac{m}{s}\right\rfloor, \left\lceil\frac{m}{s}\right\rceil \right\}^s 
  \text{ with }  \sum_{j=1}^{s} k_j = m\,,
\]
all cells of size $\prod_{j=1}^s b^{m - k_j}$ contain exactly 1 point.
Note that unless $s$ divides $m$, the number of intervals along each axis differs by at most a factor of $b$.
A second expendable technique bounded the difference of any two numbers of selected rows $k_j$ by a
value $m' \in \{0, \ldots, m\}$. Hence, $m' = 0$ realizes property A of the Sobol' sequence~\cite{Sobol67},
where contiguous blocks of $b^s$ samples are stratified, and $m' = 1$ amounts to generalized
stratification as just characterized before.

Fig.~\ref{Fig:Overview} summarizes our set of stratification constraints that are
specified by using matrices composited from generator matrices as
described before.
Regarding $(t,s)$-sequences, verifying the $(t,m,s)$-net properties for all $m \geq t$ 
is equivalent to checking whether all top-left square sub-matrices generate $(t,m,s)$-nets.
This result is stated in a slightly different manner by

\begin{theorem}[see {\cite[p.~73]{Niederreiter1992}}] \label{Thm:LinearEquations}
      Suppose that the integer $t \geq 0$ satisfies the following property: 
      For any integers $m > t$ and $k_1,...,k_s \geq 0$ with $ \sum_{j=1}^s k_j = m - t $
      and any $e_i^{(j)} \in \mathbb{F}_b$, the system of $m - t$ linear equations
      \[
        \sum_{r=0}^{m-1} c_{i,r}^{(j)} z_r = e_i^{(j)} \text{ for } 1 \leq i \leq k_j,\, 1 \leq j \leq s
      \]
      in the unknowns $z_0, \dots, z_{m-1}$ over $\mathbb{R}$ has exactly $b^t$ solutions. Then the sequence $\vec x_i$ generated by Equation~\eqref{Eqn:Algorithm} is a $(t,s)$-sequence in base $b$.
\end{theorem}

Especially $(t,s)$-sequences in base $b= 2$, as for example the construction by
Sobol'~\cite{Sobol67} are popular, because they can be computed
very efficiently using bit-parallel vector operations for the implementation
of the Galois field $\mathbb{F}_2$. The small basis, however, may
be limiting, as we claim in

\begin{theorem} \label{Thm:Pairs}
    There are no more than $b$ generator matrices in base $b$ such that all pairs of matrices generate a $(0,2)$-sequence.
\end{theorem}

A constraint graph $G$ is a graph with a vertex for each dimension, and with an edge between two dimensions that can be found in the same sequence constraint.
As $(0,2)$-sequences constraints between all pairs of dimension are included in $(t,s)$-sequence constraints, Theorem~\ref{Thm:Pairs} states that $G$ admitting not having any clique of size greater than $b$ and thus admitting a $b$ coloring is a necessary condition for the system of constraints to have solutions. 
To prove the theorem, we need the following

\begin{lemma} \label{lemma:vectors}
    In base $b$ there are no more than $b$ vectors of dimension $2$ with a non-zero first component such that all vectors are pair-wise linearly independent.
\end{lemma}

\begin{proof}
    Without loss of generality, we can assume that the first component of our vectors is always $1$ as any vector is linearly equivalent to one with the first component equal to $1$.
    For all $k_1 \neq k_2$ the vectors $(1,k_0)$ and $(1,k_1)$ are linearly independent.
    There are $b$ possible values for the second component.
    Thus there are no more than $b$ vectors of dimension $2$ with a non-zero first component such that all vectors are pair-wise linearly independent.
\end{proof}

With the lemma at hand, we are ready to prove Theorem~\ref{Thm:Pairs}:

\begin{proof}
    Let $C_1, \dots, C_s$ be $s$ generator matrices in base $b$ such that all pairs of matrices generate a $(0,2)$-sequence.
    Then all pairs of top-left cornered square submatrices of size $m$ generate a $(0,m,2)$-net.
    For $m=1$ this means for $1 \leq j \leq s$ all $c^{(j)}_{1,1} \neq 0$.
    For $m=2$ this means that all the first rows of each matrix must be pair-wise linearly independent.
    By Lemma~\ref{lemma:vectors}, there are no more than $b$ linearly independent first rows.
    Thus there are no more than $b$ matrices in base $b$ such that all pairs of matrices generate a $(0,2)$-sequence.
\end{proof}

\begin{figure}
\centering
\begin{tabular}{ccc}
\begin{minipage}{0.45\linewidth}
\begin{tcolorbox}[width=\linewidth,colback=backcolour,boxsep=-2mm,colframe=backcolour,arc=0mm]
\begin{lstlisting}[style=mystyle2]
#One-weak-constraint
s=6
m=10
b=3
weak 1 net 0 1 2 3 4 5
\end{lstlisting}
\end{tcolorbox}
\begin{tcolorbox}[width=\linewidth,colback=backcolour,boxsep=-2mm,colframe=backcolour,arc=0mm]
\begin{lstlisting}[style=mystyle2]
#Overlapping-constraints
s=6
m=10
b=3
net t0 0 1
net t0 1 2
weak 1 net t1 3 4 5
weak 1 net t2 0 1 2 3 4 5
\end{lstlisting}
\end{tcolorbox}
\end{minipage}
& &
\begin{minipage}{0.45\linewidth}
\begin{tcolorbox}[width=\linewidth,colback=backcolour,boxsep=-2mm,colframe=backcolour,arc=0mm]
\begin{lstlisting}[style=mystyle2]
#Generic-proj-LDS
s=6
m=10
b=3
net 0 1
net 1 2
net 2 3
net 3 4
net 4 5
weak 1 net 0 2
weak 1 net 0 3
weak 1 net 0 4
weak 1 net 0 5
weak 1 net 1 3
weak 1 net 1 4
weak 1 net 1 5
weak 1 net 2 4
weak 1 net 2 5
weak 1 net 3 5
\end{lstlisting}
\end{tcolorbox}
\end{minipage}
\end{tabular}
\caption{Example profiles to specify generator matrices by constraints: The parameter
  $s$ refers to the problem dimension, $b$ is the prime base, and $m$ is the
  matrix size to generate up to $b^m$ points. The remaining lines
  specify net and hence implicitly stratification constraints on selected
  dimensions.
  The keyword \texttt{net} requires $(t,m,s)$-net properties on the subsequent set of dimensions,
  where by default $t = 0$. Larger $t$-parameters are specified explicitly using the keyword \texttt{t},
  e.g. \texttt{t2} $\equiv t = 2$.
  The keyword \texttt{weak} indicates that the constraint
  is not strict, and its relative strength $w_{\nu_j}$ is given by the value next to the keyword.}
  \label{Fig:Profiles}
\end{figure}

\section{Specifying Generator Matrices by Constraints}

The profile language (see Fig.~\ref{Fig:Profiles}) as introduced in~\cite{matbuilder2022} allows one
to specify generator matrices by constraints.
A profile first selects the dimension $s$, prime base $b$, and matrix
size $m$, followed by a list of constraints. Each constraint affects
a selected subset of dimensions, which allows for forging
application-specific uniformity properties. Then the constraints
are transformed into an Integer Linear Program (see Fig.~\ref{Fig:ILP}) whose
solution specifies the $s$ generator matrices of size $m\times m$.

We extend the syntax of the profile language to include $(t,m,s)-$net
constraints with arbitrary $t\geq0$ as an additional parameter to the
\texttt{net} and \texttt{weak net} keywords. As an example, the profile \texttt{Overlapping-constraints} in
Fig.~\ref{Fig:Profiles} shows the constraints to
design a progressive net in dimension 6 and base 3 up to
$3^{10}$ points that is a $(0,10,2)$-net for the first two pairs of
dimensions, as close to a $(1,10,3)$-net as possible for the dimensions
$\{3,4,5\}$, and as close to a $(2,10,6)$-net as possible.

\subsection{Overlapping Constraints}

Theoretical results~\cite{Niederreiter1992} state that the best possible $t$-parameter for a (t,m,s)-net is $t = s - b - 1$.
However, many applications exhibit a structure where the uniformity of a subset of dimensions greatly impacts the integration error.
This has already been observed in the work of Joe and Kuo~\cite{joe2008constructing} and L'Ecuyer et al.~\cite{l2022tool}.
As the dimension of the subsets often is much smaller than the global dimension, the $t$-parameter of the subset of dimensions theoretically can be smaller than the global $t$-parameter.
In our profile language (see Fig.~\ref{Fig:Profiles}), it is straightforward to specify such potentially overlapping constraints, which offers a new way of designing generator matrices.

\subsection{Specifying Constraints in $\mathbb{Z}$}

Solvers are usually devised to work in $\mathbb{Z}$, however, our constraints are in $\mathbb{F}_b$. In order to take advantage of existing  solvers, we  need to convert our constraints to $\mathbb{Z}$.
In $\mathbb{F}_b$ our constraints are of the form $\sum_i w_i x_i \neq 0$.
We transform them to $0 < \sum_i w_i x_i + kb < b$ with $k \in \mathbb{Z}$ an additional variable to solve for and $0 \leq x_i < b$  (see Fig.~\ref{Fig:ILP}).
This way we emulate the modulo arithmetic of prime base Galois fields.
Note that this trick is limited to prime bases $b$.

\subsection{Weak Constraints}

Each net constraint on a set of dimensions is actually a set of linear sub-constraints, one per elementary interval shape.
It often happens that constraints cannot be satisfied.
This can be due to a too small $t$-parameter or because of conflicting overlapping constraints.
To alleviate the issue, we introduce weak constraints. 
All constraints were crafted such that the determinant of the corresponding matrix $M$ is strictly in the range $\{1,\ldots, b-1\}$.
For the $j$-th weak constraint, we now make these bounds depend on a variable $\nu_j \in [0,1]$ in the following manner:
\[
  \nu_j \leq \sum_i w_i x_i + k_jb \leq (b-1)\nu_j\,.
\]
This way, if $\nu_j = 0$ the constraint becomes $\sum_i w_i x_i + k_jb = 0$ and if $\nu_j = 1$ the constraint results to be $0 < \sum_i w_i x_i + k_jb < b$ which is equivalent to the original hard constraint.
Maximizing the sum $\sum w_{\nu_j}\nu_j$, where $w_{\nu_j}$ is the weight of constraint $j$, thus maximizes the number of satisfied sub-constraints while allowing some of them not to be satisfied in case they are infeasible.
The complete setup of the Integer Linear Program is summarized in Fig.~\ref{Fig:ILP}.

\subsection{Polynomial Integer Program}

A matrix $A \in \mathbb{F}_b^{(m-t)\times m}$ has full rank if and only if at least one of the square sub-matrices $A_i$ obtained by dropping any $t$ columns $i = (i_1, \dots, i_t)$ of $A$ has full rank.
This means that at least one such sub-matrix has a non-zero determinant.
This can be represented by the set of polynomial constraints of degree $m-t$, 
$$\exists i \in \{(i_1, \dots, i_t) \in\mathbb{N}^t\mid 1\leq i_1 < \dots < i_t \leq m\}: \text{ det}(A_i) \neq 0\,.$$

Imposing all $M_\mathbf{k}$ to have full rank can thus be expressed as satisfying a set of polynomial constraints on values of matrices $C_j $.
\[
      \begin{array}{c} 
        \forall \mathbf{k} := (k_1, \dots, k_s) \in \mathbb{N}^s \text{ with } \sum_{j=1}^s k_j = m-t \\
        \exists i \in \{(i_1, \dots, i_t) \in\mathbb{N}^t\mid 1\leq i_1 < \dots < i_t \leq m\}: \text{ det}(M_{\mathbf{k},i}) \neq 0\,.
      \end{array}
\]
with $M_{\mathbf{k},i}$ denoting the matrix $M_\mathbf{k}$ with the set of columns $i$ removed. 
The set of matrices $C_1, \ldots, C_s$ describes a (t,m,s)-net if and only if it satisfies a set of polynomial constraints stating that all $M_\mathbf{k}$ constructed from these matrices (see Sec.~\ref{Sec:Constraints}) have full rank.

\subsection{Integer Linear Programs}

Theorem~\ref{Thm:LinearEquations} advocates linear constraint solving, which, however, is NP-complete.
Unless $P = NP$, the computational complexity is exponential
in the number of variables $s m^2$.
In practice, such algorithms are infeasible, as the number of variables is too high.
However, assuming all values in the matrices to be known up to column $i-1$ and trying to determine the $i$-th column simplifies the problem to that of satisfying a set of linear constraints with $s m$ variables, since the determinant is a linear function of the matrix columns.
In combination with the heuristic applied in modern Integer Linear Program Solvers~\cite{CPLEX},
it becomes possible to compute tangible results.

We hence use a greedy algorithm to construct the generator matrices column-by-column:
\[
    \begin{pNiceMatrix}[margin]
    \Block[fill=green!35]{1-1}{}
    c_{1,1} & \cdots & c_{1,m} \\
    \Block[fill=myblue!35]{2-1}{}
    \vdots & \ddots & \vdots \\
    c_{m,1} & \cdots & c_{m,m} \\
    \end{pNiceMatrix}
    \quad \rightarrow \quad
    \begin{pNiceMatrix}[margin]
    \Block[fill=green!15]{2-2}{}
    c_{1,1} & \Block[fill=green!35]{2-1}{} \cdots & c_{1,m} \\
    \vdots & \ddots & \vdots \\
    c_{m,1} & \Block[fill=myblue!35]{1-1}{} \cdots & c_{m,m} \\
    \end{pNiceMatrix}
    \quad \rightarrow \quad
    \begin{pNiceMatrix}[margin]
    \Block[fill=green!15]{3-3}{}
    c_{1,1} & \cdots & \Block[fill=green!35]{3-1}{} c_{1,m} \\
    \vdots & \ddots & \vdots \\
    c_{m,1} & \cdots & c_{m,m} \\
    \end{pNiceMatrix}
\]
At each step, the top-left square sub-matrices are checked to
guarantee the sequence property (highlighted in green).
In the highlighted columns, elements in green are chosen at random according to the constraints,
while the blue elements are not constrained and become selected at random.
As our matrices have finite size, Definition~\ref{Def:tsSequence}
of a $(t,s)$-sequence is only ensured up to $b^m$ points. Hence, we call
point sets generated by such matrices \emph{progressive $(t,m,s)$-nets}.

To determine these linear constraints for each $M_\mathbf{k}$, we check whether the matrices have full rank by performing a \textit{symbolic} Gaussian elimination in the Galois field $\mathbb{F}_b$. Gaussian elimination seeks to triangularize a matrix by iteratively subtracting linear combinations of rows. For a rectangular matrix, the Gaussian elimination results in 3 possible outcomes:
\[
\begin{array}{ccc} %
\begin{pNiceMatrix}[margin]
*          & \Cdots & \Cdots & \Cdots & \Cdots & *         & v_1\\
0         & \Ddots &             & &            & \Vdots & \\
\Vdots & \Ddots & *          & \Cdots & \Cdots & *          & \\
\Vdots         &            & 0          & \Cdots & \Cdots & 0         &  \Vdots \\
0         & \Cdots & \Cdots & \Cdots & \Cdots & 0         & v_{m - t}
\end{pNiceMatrix}
\; & \;
\begin{pNiceMatrix}[margin]
*          & & \Cdots & \Cdots & \Cdots & *         & v_1\\
0         & \Ddots  &            &          &  & \Vdots          & \\
\Vdots & \Ddots  & \Ddots &           &  & \Vdots         &  \Vdots \\
\Vdots &             & \Ddots &  \Ddots         &  & \Vdots & \Vdots \\
0         & \Cdots  & \Cdots & 0         & * & *         & v_{m - t}
\end{pNiceMatrix}
\; & \;
\begin{pNiceMatrix}[margin]
*          & & \Cdots & \Cdots & \Cdots & *         & v_1\\
0         & \Ddots  &            &          &  & \Vdots          & \\
\Vdots & \Ddots  & \Ddots &           &  & \Vdots         &  \Vdots \\
\Vdots &             & \Ddots &  *         & \Cdots & * & \Vdots \\
0         & \Cdots  & \Cdots & 0         & \Cdots & 0         & v_{m - t}
\end{pNiceMatrix}
\\
(a) & (b) & (c)
\end{array}
\]
In case (a), no choice of $v_j$ can make the matrix full rank since at least two rows consist of zeros, possibly except for their last component $v_{m-t}$ and $v_{m-t-1}$, and hence these rows must be linearly dependent.
In case (b), the matrix necessarily has full rank regardless of the last column $v$  as the row vectors are linearly independent even without their last component.
In case (c), the matrix has full rank if and only if $v_{m-t} \neq 0$ to ensure the last row is not identically zero.
Following the Gaussian elimination process, $v_{m-t}$ is a linear combination of the variables of the last column of $M_\mathbf{k}$ with weights depending on the values of the first columns of $M_\mathbf{k}$.

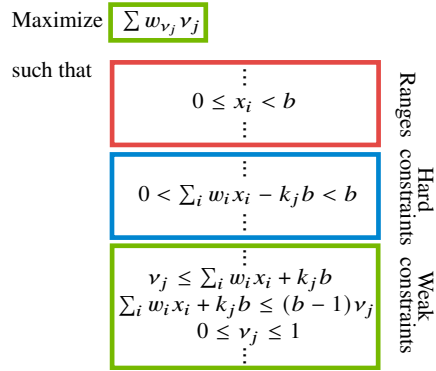
\begin{figure}\centering
  \begin{overpic}[width=6cm]{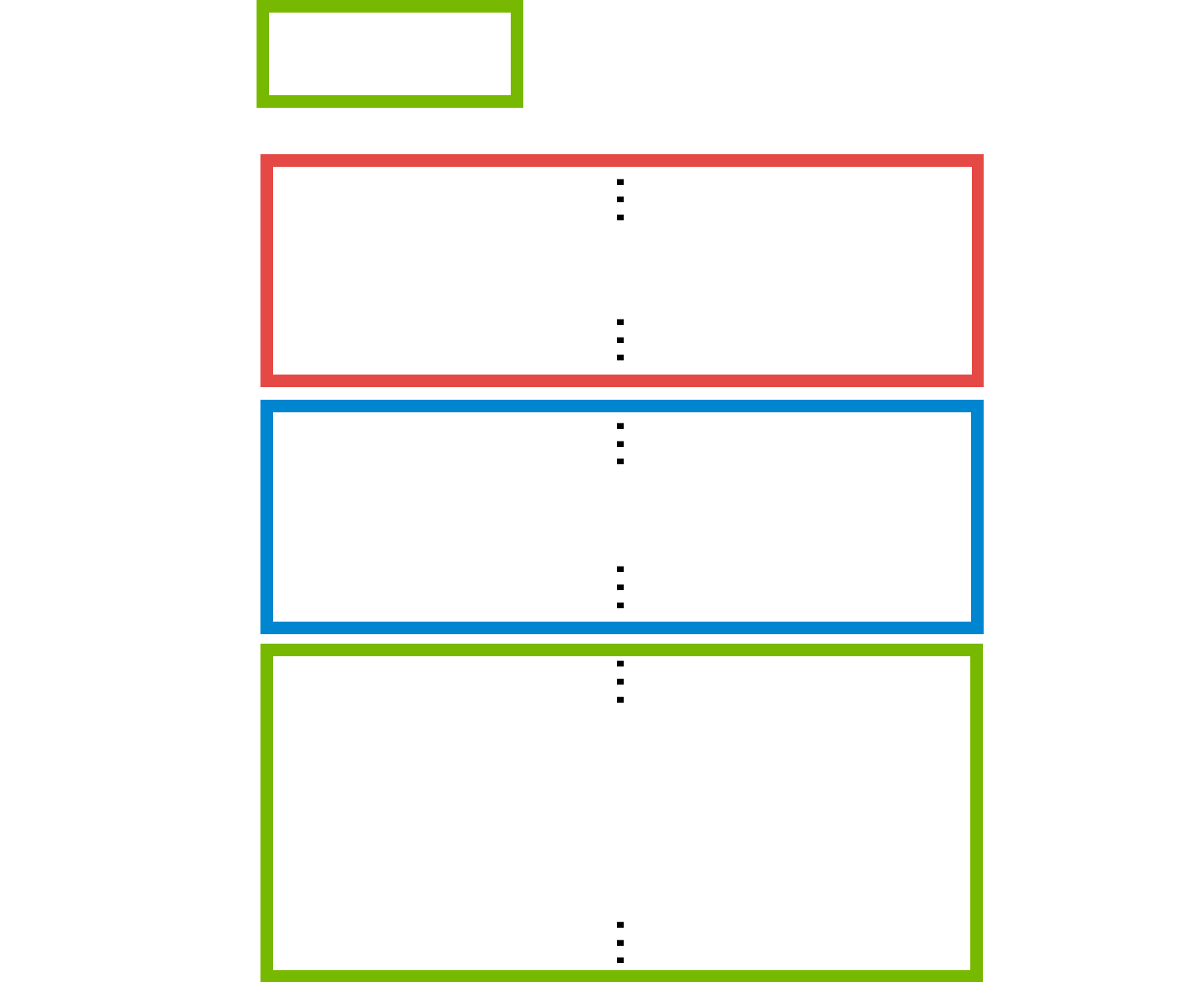}\small
    \put(0,76){Maximize}
    \put(24.5,76){$\sum w_{\nu_j} \nu_j$}
    \put(0,65){such that}
    \put(90,45){\rotatebox{270}{Hard}}
    \put(86,49){\rotatebox{270}{constraints}}
    \put(90,22){\rotatebox{270}{Weak}}
    \put(86,25){\rotatebox{270}{constraints}}
    \put(85,66){\rotatebox{270}{Ranges}}
    \put(28,37.5){$0 < \sum_i w_i x_i - k_jb < b$}
    \put(30,19){$\nu_j \leq \sum_i w_i x_i + k_jb$}
    \put(23.5,13){$\sum_i w_i x_i + k_jb \leq (b-1)\nu_j$}
    \put(41,7){$ 0\leq \nu_j \leq 1$}
    \put(40,58){$0 \leq x_i < b$}
  \end{overpic}
    \caption{Anatomy of an Integer Linear Program (ILP). Please refer
      to the text for details.}
      \label{Fig:ILP}
\end{figure}

In summary, to
grow the matrices $C_i$ according to our greedy strategy, 
the values of the $c^{(i)}_{l,m+1}$ (abstracted as $x_i$ in our formulas) in the last column of
their respective $C_i$ are determined by solving an Integer Linear Program,
which consists of an objective function to maximize subject
to a set of constraints. Fig.~\ref{Fig:ILP} shows the anatomy of our Integer Linear Programs:
The range constraints enforce that $c^{(i)}_{l,m+1} \in \{0, \ldots, b -1\}$
and the hard uniformity constraints enforce a non-zero
determinant
to guarantee the design constraints of stratification,
net, and sequence properties as introduced in Sec.~\ref{Sec:Constraints}.
Remember that matrices $M_\mathbf{k}$ are constructed from %
the first rows of the $C_i$ matrices and hence include some of the $c^{(i)}_{l,m+1}$.
Indicated by $\nu_j = 1$, a satisfied weak constraint adds its
weight $w_{\nu_j}$ to the objective function. Otherwise, a zero linear combination (stating that the corresponding $M_\mathbf{k}$ is not full rank) %
comes along with $\nu_j = 0$.

\section{Results} \label{Sec:Results}

In our previous work~\cite{matbuilder2022}, we focused on
computer graphics applications including image synthesis, 
parametric texture exploration, and optimal control.
At that time, we only supported the special cases of stratification and $(0,m,s)$-net properties,
where $t=0$. For our new results, we investigate the larger solution space for generator matrices
provided by $t\geq0$ and weak constraints. The implementation of
the MatBuilder software is publicly available at \url{https://github.com/loispaulin/matbuilder}.

Fig.~\ref{fig:full6d} shows an initial experiment in $s = 6$ dimensions.
We observe that both MatBuilder and LatNetBuilder~\cite{latnetbuilder}
achieve good performance in terms of discrepancy when maximizing
the uniformity. Note that the \texttt{One-weak-constraint}
profile in Fig.~\ref{Fig:Profiles} causes the solver to approximate a progressive $(0, m, 6)$-net
in base $b=3$ as closely as possible, while theoretically it does not exist.

\begin{figure} 
  \centering
  \includegraphics[height=0.4\linewidth]{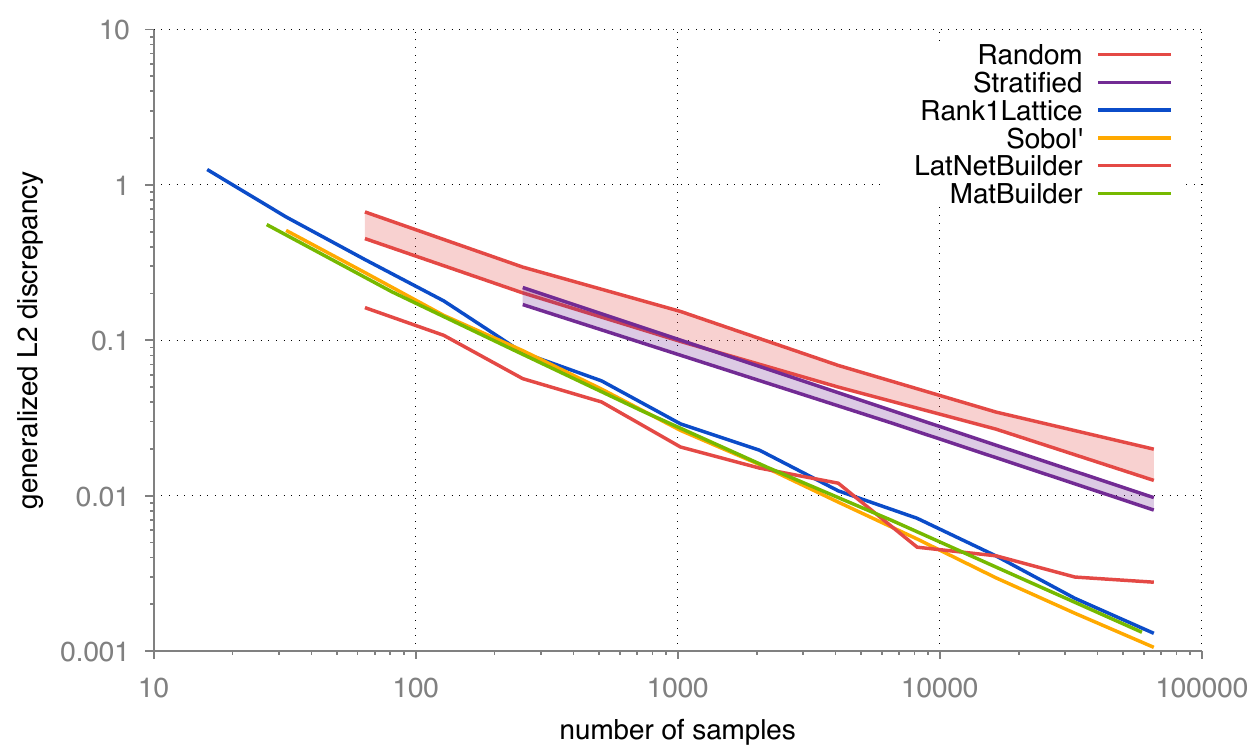}
  \caption{We measure uniformity in terms generalized
      $l_2$-discrepancy in dimension $s=6$. The MatBuilder ($b = 3$) results are
      obtained using a \texttt{"weak 1 net 0 1 2 3 4 5"} constraint in the \texttt{One-weak-constraint} profile in Fig.~\ref{Fig:Profiles}.
      For LatNetBuilder ($b = 2$), we have used a figure of merit
      minimizing the discrepancy. Rank1Lattice refers to \cite{keller2004stratification}.
      For reference, stratified sampling partitions each dimension into the same amount of intervals and randomly samples once inside each resulting hypercube.
      The range for random and stratified sampling results from 64 independent realizations.}
  \label{fig:full6d}
\end{figure}

\subsection{Overlapping Net Constraints}

In Fig.~\ref{Fig:Projective}, we present
results for the \texttt{Generic-proj-LDS} profile in Fig.~\ref{Fig:Profiles} which ensures
progressive $(0,m,2)$-net properties for consecutive pairs of dimensions. For all other pairs
of dimensions, the additional weak constraints ask the solver to establish a
progressive $(0,m,2)$-net  property if possible. Using weak
constraints in the profile maximizes the number of elementary intervals
checking their part of the progressive $(0,m,2)$-net property.
This results in points that, even though technically not a progressive $(0,m,2)$-net, exhibit a similar quality in terms of discrepancy.
As compared to the
Sobol' sequence, the constraint based generator matrices
clearly improve the quality across the 2D
sample projections. Converting these constraints into a
loss function for a stochastic matrix construction,
LatNetBuilder~\cite{latnetbuilder} does not achieve a comparable
quality in the projections on its own. Yet, a constraint based specification
may help stochastic optimization~\cite{matbuilder2022}.

\begin{figure} %
  \centering
  \begin{tabular}{ccc}
  \includegraphics[height=0.32\linewidth]{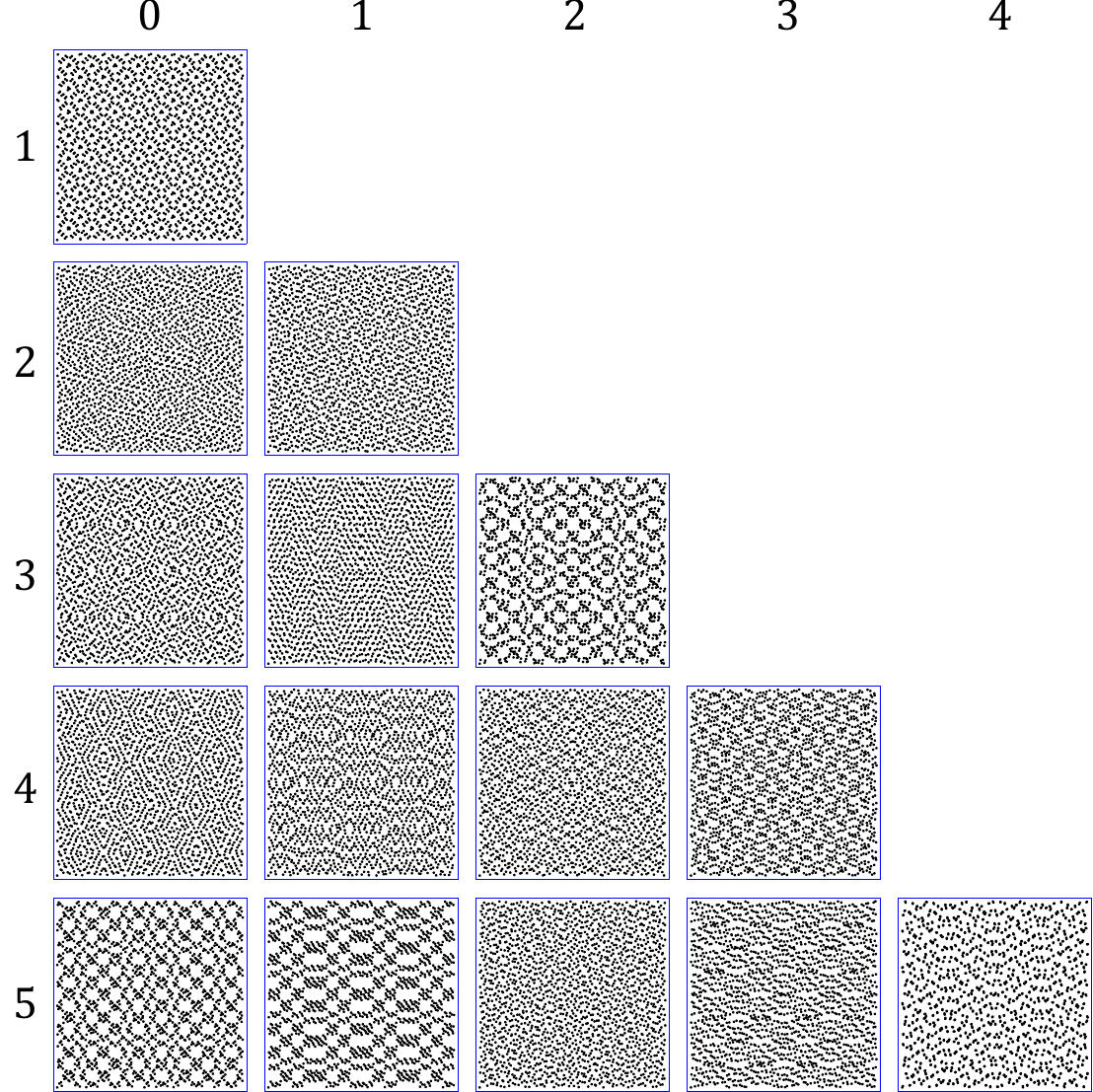} &
  \includegraphics[height=0.32\linewidth]{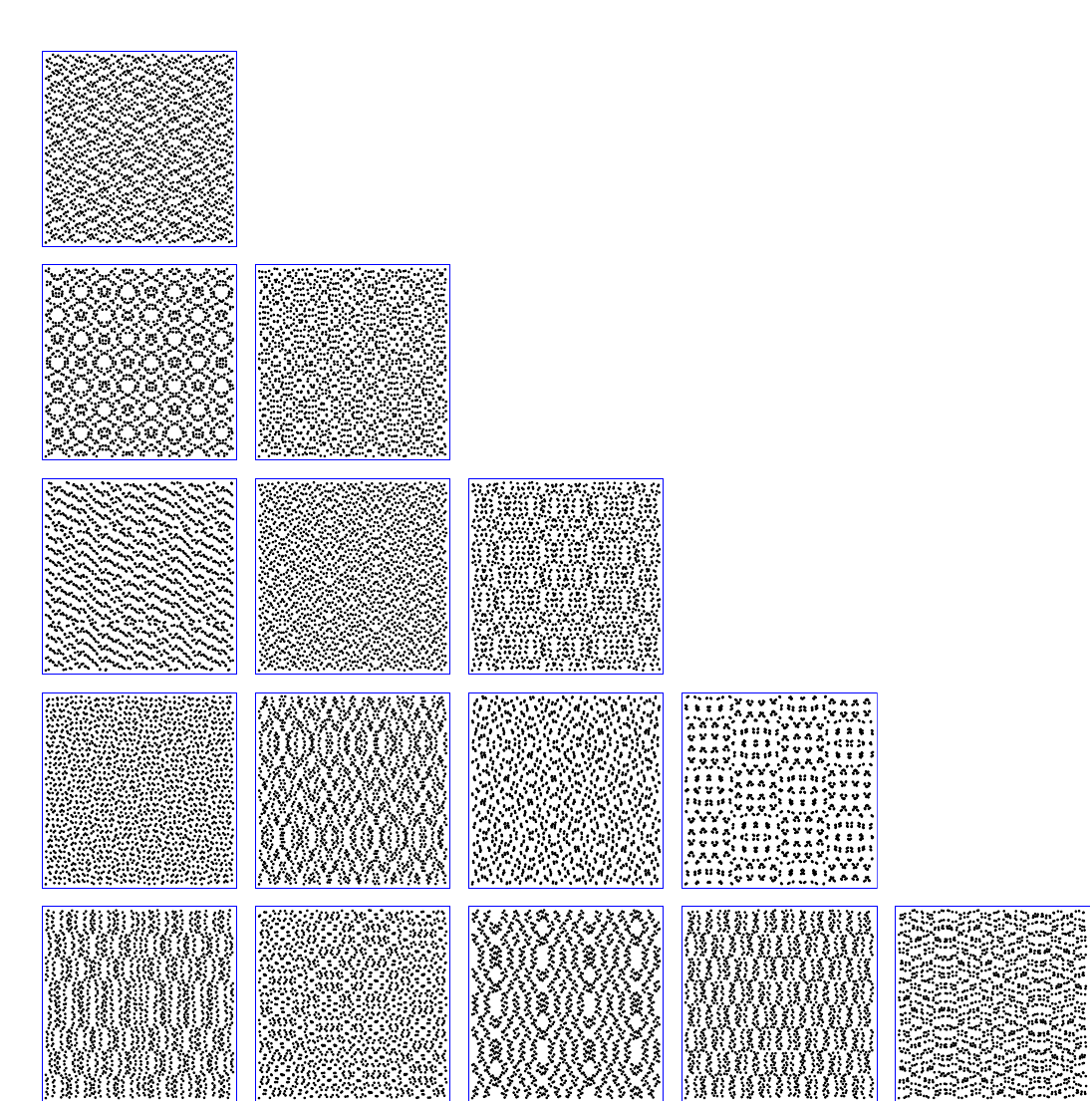}  &
  \includegraphics[height=0.32\linewidth]{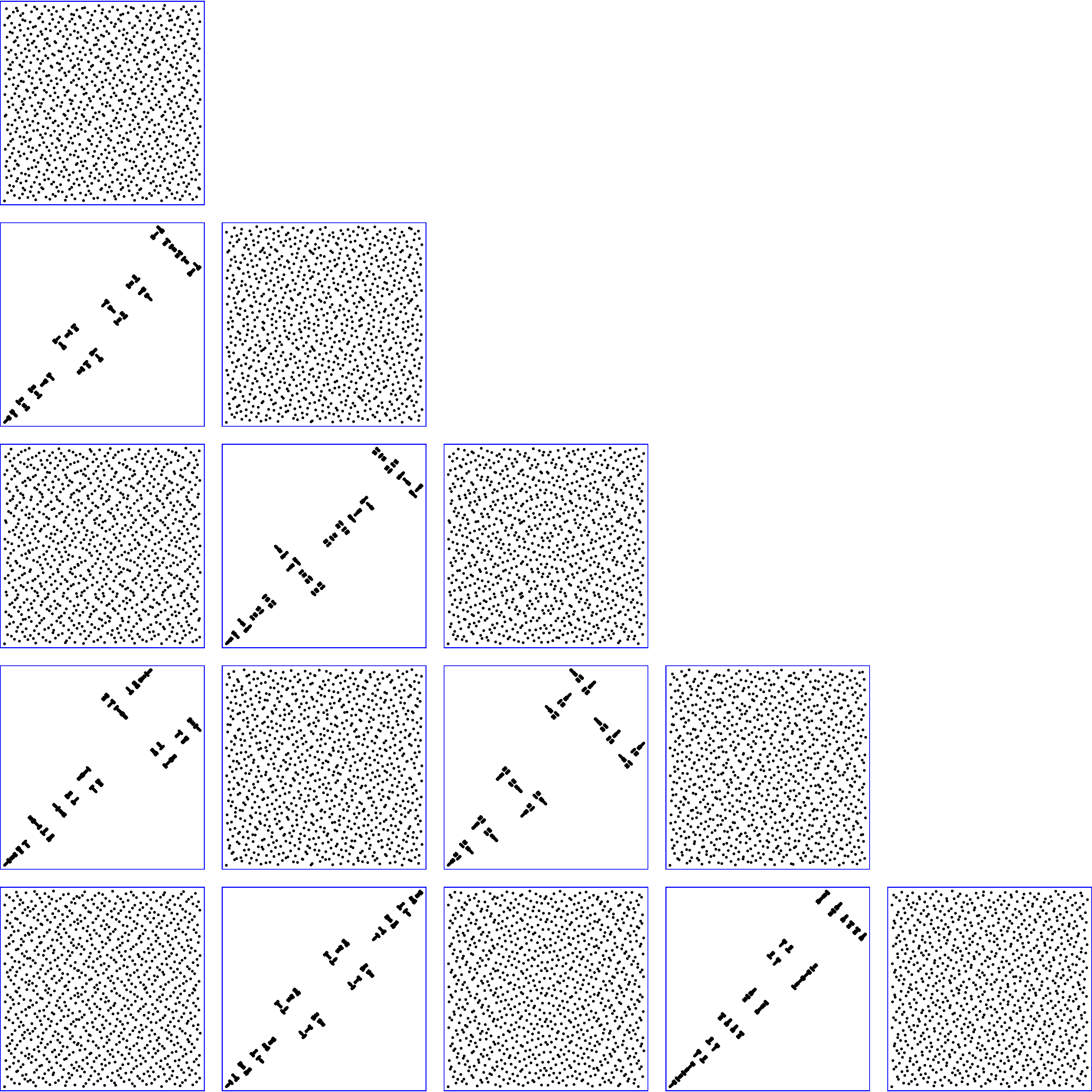}\\
  Sobol' sequence, $b = 2$ &  LatNetBuilder, $b = 2$ & MatBuilder, $b = 2$ \\
  \end{tabular}
  \begin{tabular}{cc}
  \raisebox{.7cm}{\includegraphics[height=0.32\linewidth]{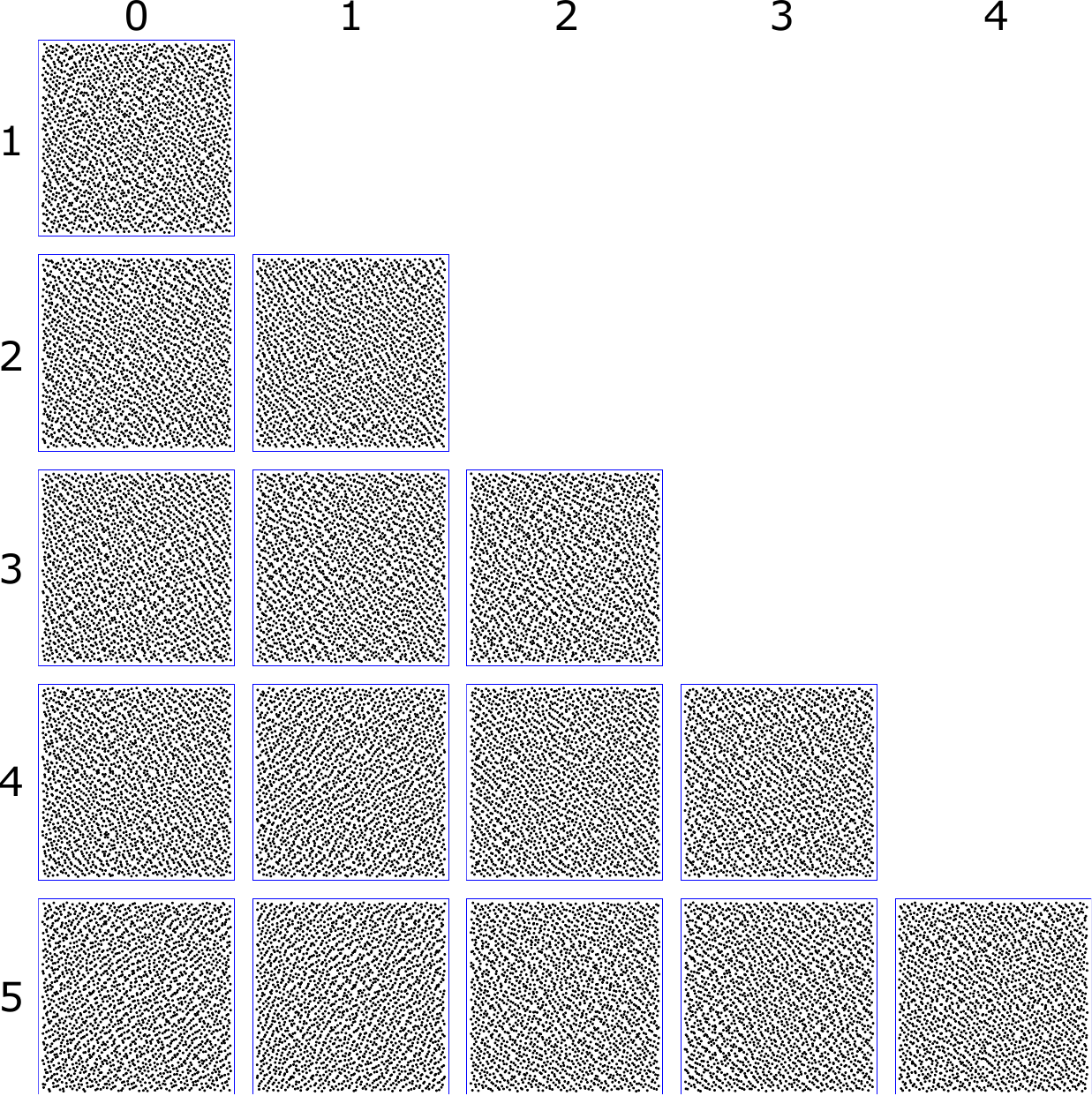}} & \includegraphics[height=0.4\linewidth]{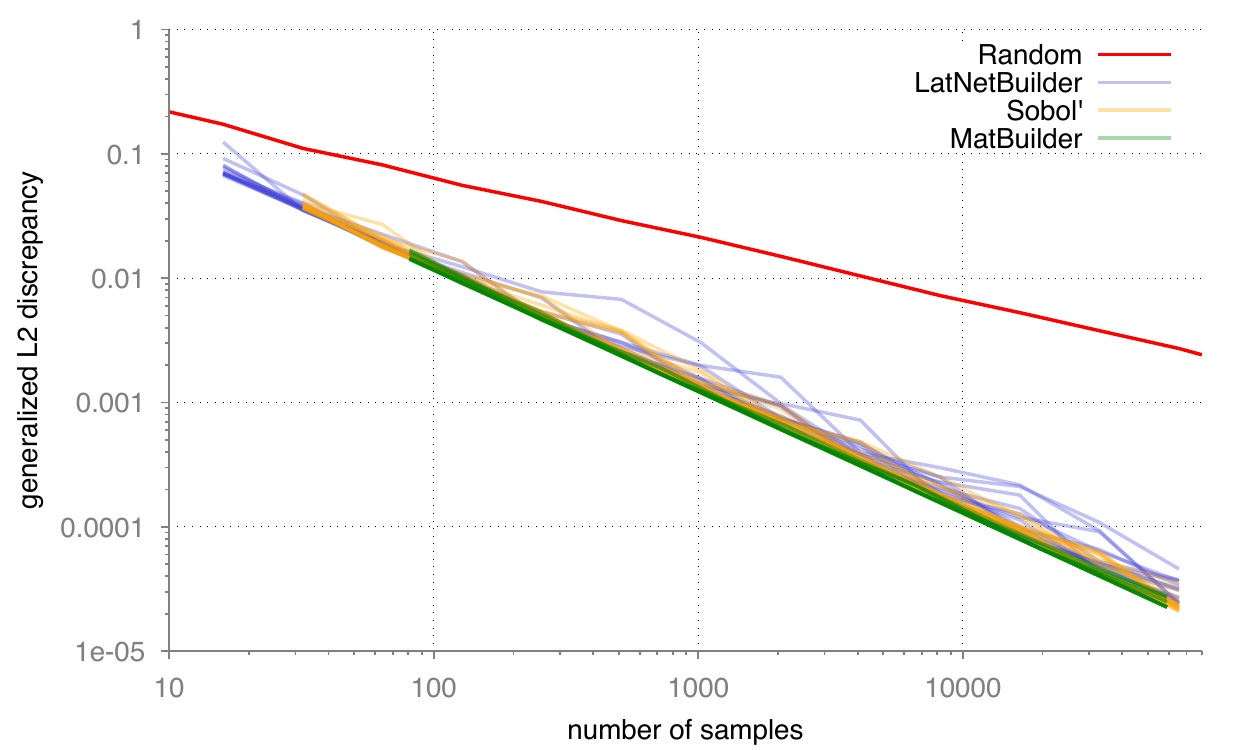}\\
  MatBuilder, $b = 3$ & Discrepancy of two-dimensional projections
  \end{tabular}
  \caption{Comparison of two-dimensional projections. For 2048 points generated in base $b = 2$, both the Sobol' sequence and
  the result from LatNetBuilder show the known typical patterning. Both do not fulfill
  $(0,m,2)$-net constraints on pairs of consecutive dimensions. Given the \texttt{Generic-proj-LDS} profile in Fig.~\ref{Fig:Profiles},
  MatBuilder can enforce such constraints. Then, off-diagonal projections may lack uniformity as a consequence of Theorem~\ref{Thm:Pairs}.
  However, using the same profile for $b = 3$ and a similar number of 2187 points,  MatBuilder finds high quality matrices
  that satisfy both the hard and weak constraints. The graph simultaneously plots the generalized $l_2$-discrepancy
  of all two-dimensional projections for random sampling, the Sobol' sequence, the LatNetBuilder result,
  and the MatBuilder result in $b = 3$. The generator matrices specified by constraints
  consistently generate points of excellent low discrepancy with the least variation across all projections.}
  \label{Fig:Projective}
\end{figure}

\subsection{Playing with $t$-Parameters}

Our system empowers the user to play with $t$-parameters provided as weak constraints.
In order to satisfy a weak constraint, the greedy algorithm maximizes the number of elementary intervals of size $b^{-t}$ containing $b^t$ points, approximating the properties of $t$-parameters that are theoretically impossible.
For example, in base $b=3$ the best possible $t$-parameter for a
progressive $(t,m,6)$-net is $t=3$. However, by asking the solver to generate matrices with $t \in \{0, 1, 2\}$ as a weak constraint, we are able to improve uniformity.

In Fig.~\ref{Fig:Perf-t}, we demonstrate that low discrepancy can be achieved by weak constraints for the examples of $t \in \{0, \ldots, 4\}$ and $m$ up to 10 in $s = 6$ dimensions.
As expected, increasing the $t$-parameter
has a negative impact on the six-dimensional generalized $l_2$-discrepancy
(Fig.~\ref{Fig:Perf-t}-$a$) and the sample projection uniformity
(Fig.~\ref{Fig:Projective2}). While the construction time increases
with the matrix size $m$, the number of
constraints to satisfy decreases with increasing $t$-parameter. Hence, the smaller $t$, the
more greedy expansion steps of the matrix columns and rows need to be executed as $m$ increases (Fig.~\ref{Fig:Perf-t}-$b$).

\begin{figure}[!h]
  \subfigure[]{\includegraphics[width=0.5\textwidth]{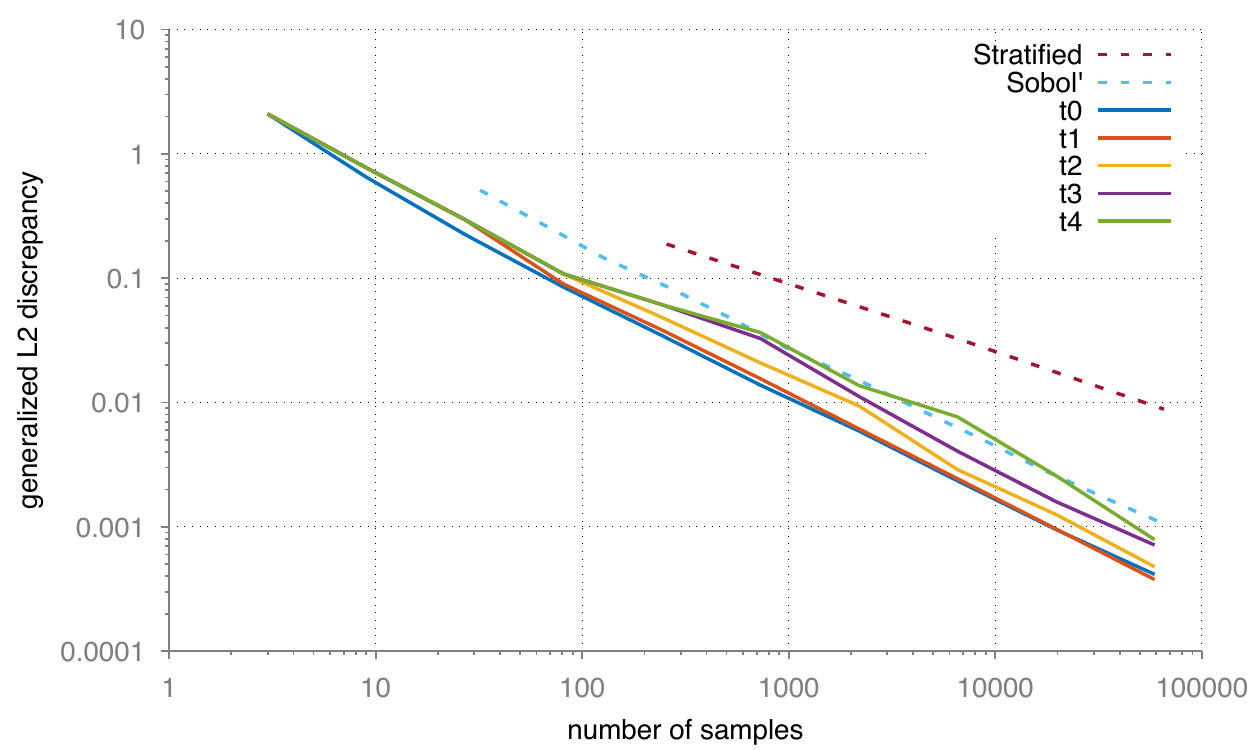}}
  \subfigure[]{\includegraphics[width=0.5\textwidth]{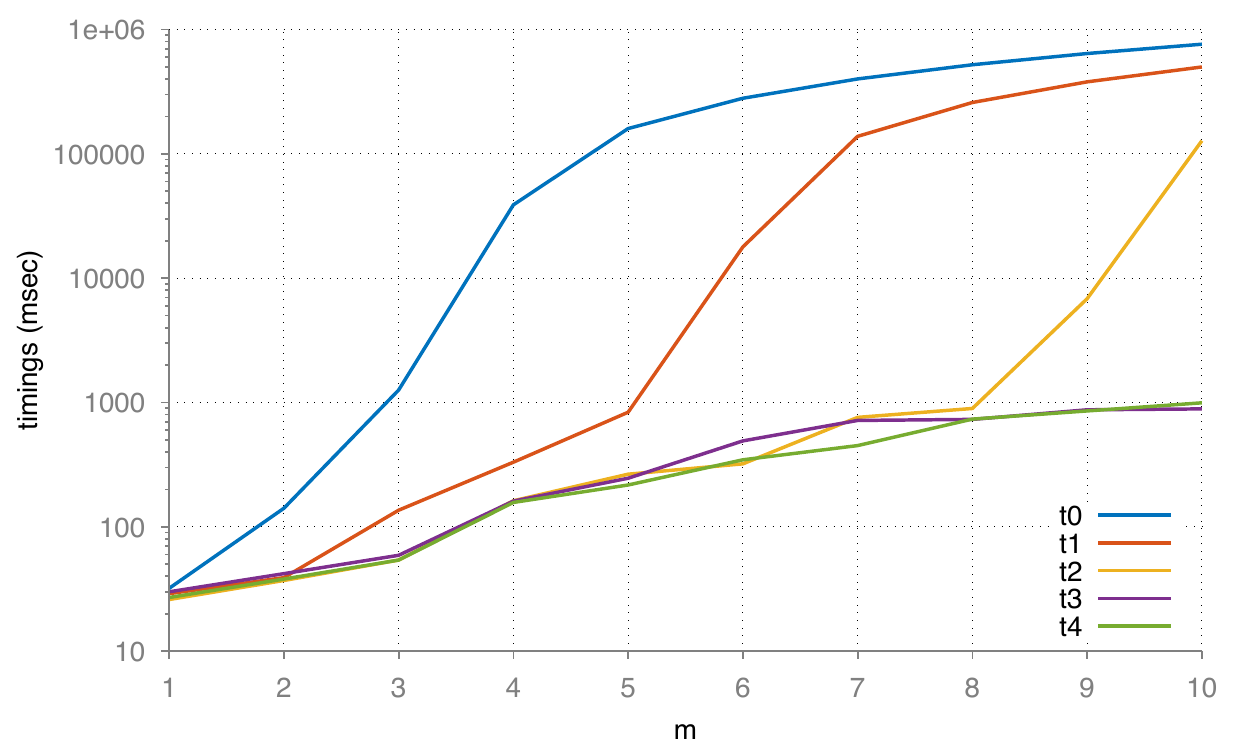}}
  \caption{Performance evaluation when increasing the
    $t$-parameter from 0 to 4 on a weak net profile in base 3 and
    dimension 6 up to $3^{10}$ points: $(a)$ Quality
    evaluation in terms of generalized $l_2$-discrepancy.  $(b)$ Timings of the
    solver as a function of the matrix size $m$.}
  \label{Fig:Perf-t}
\end{figure}

\begin{figure}[!h]\centering
  \subfigure[]{\includegraphics[width=0.32\textwidth]{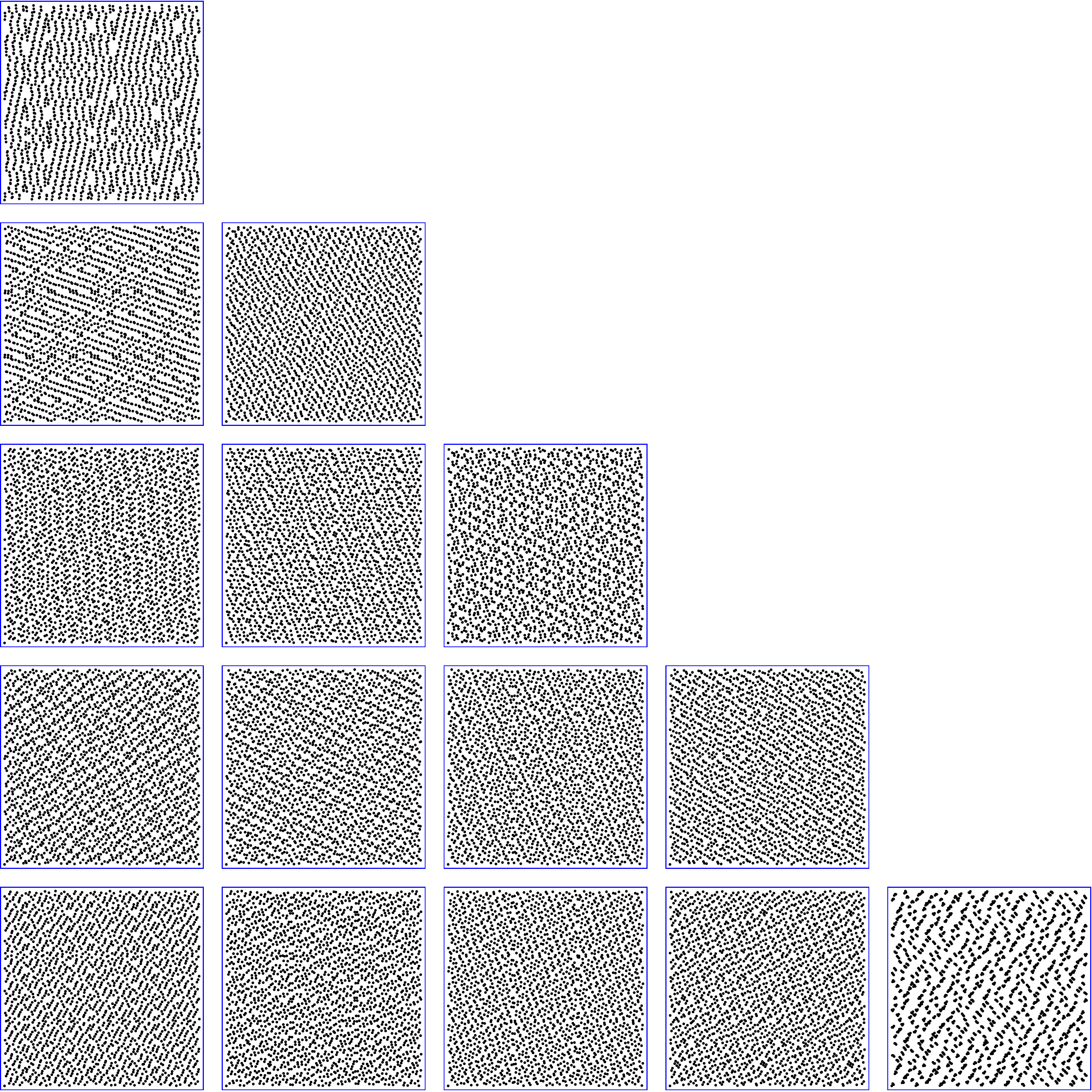}}
  \subfigure[]{\includegraphics[width=0.32\textwidth]{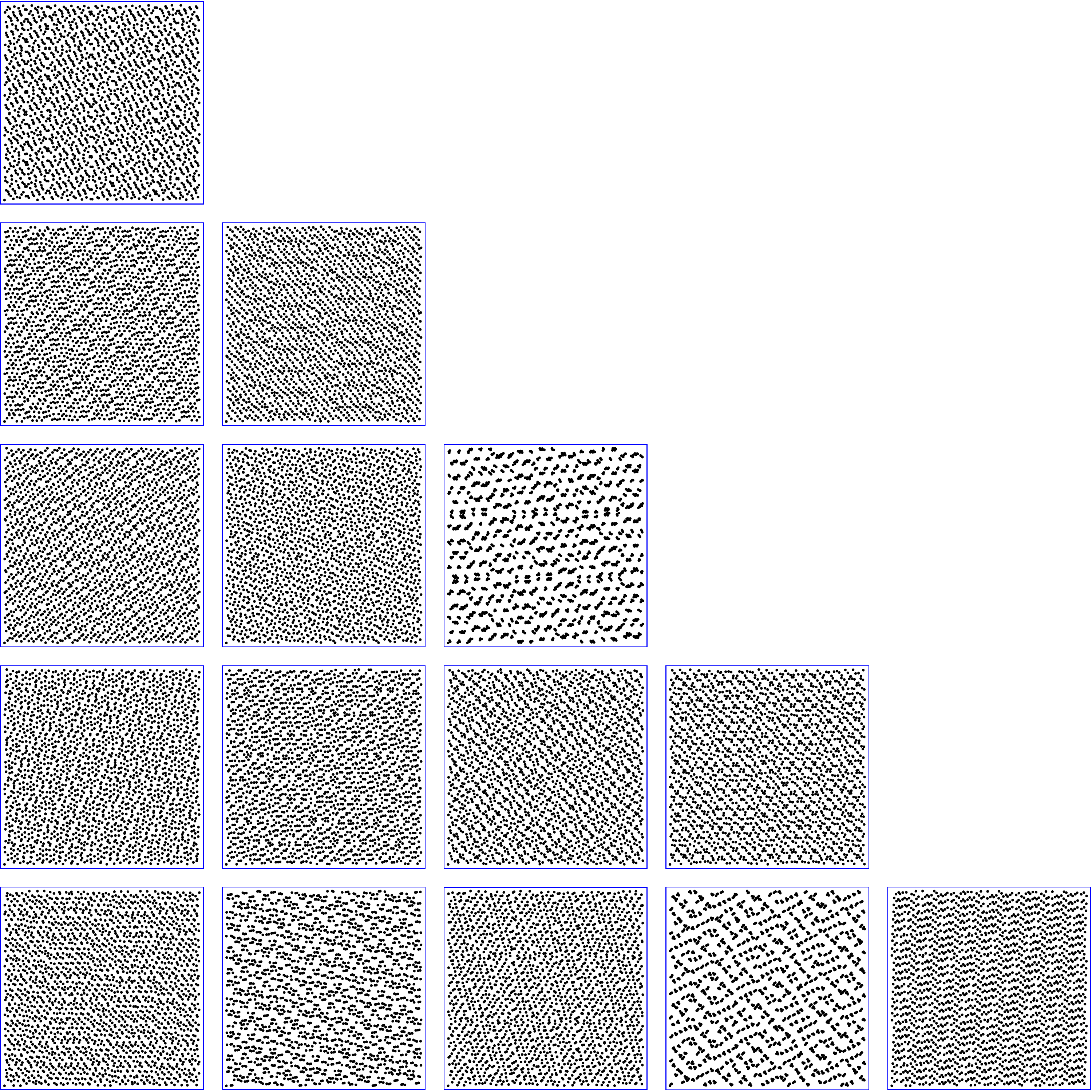}}
  \subfigure[]{\includegraphics[width=0.32\textwidth]{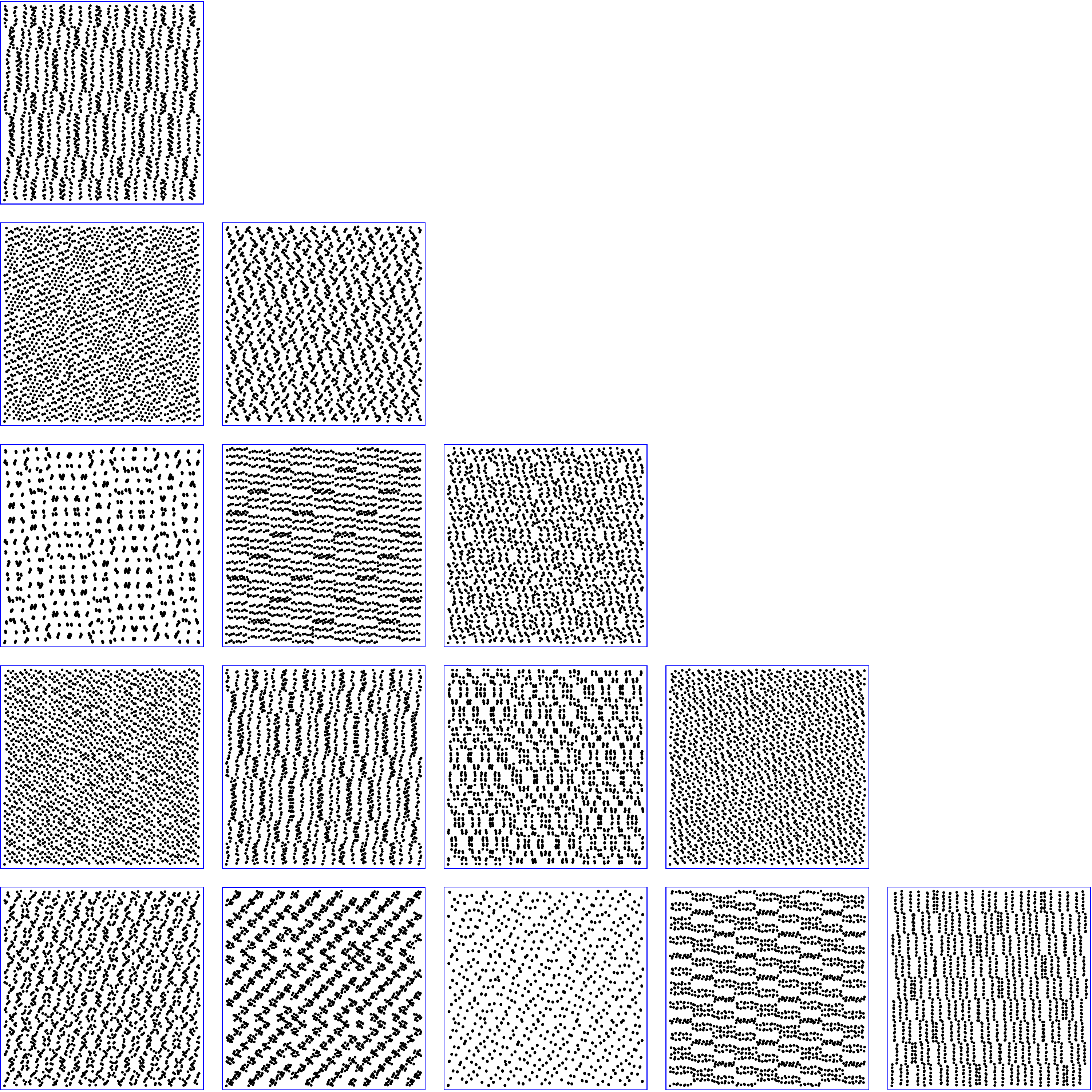}} \\
  \subfigure[]{\includegraphics[width=0.32\textwidth]{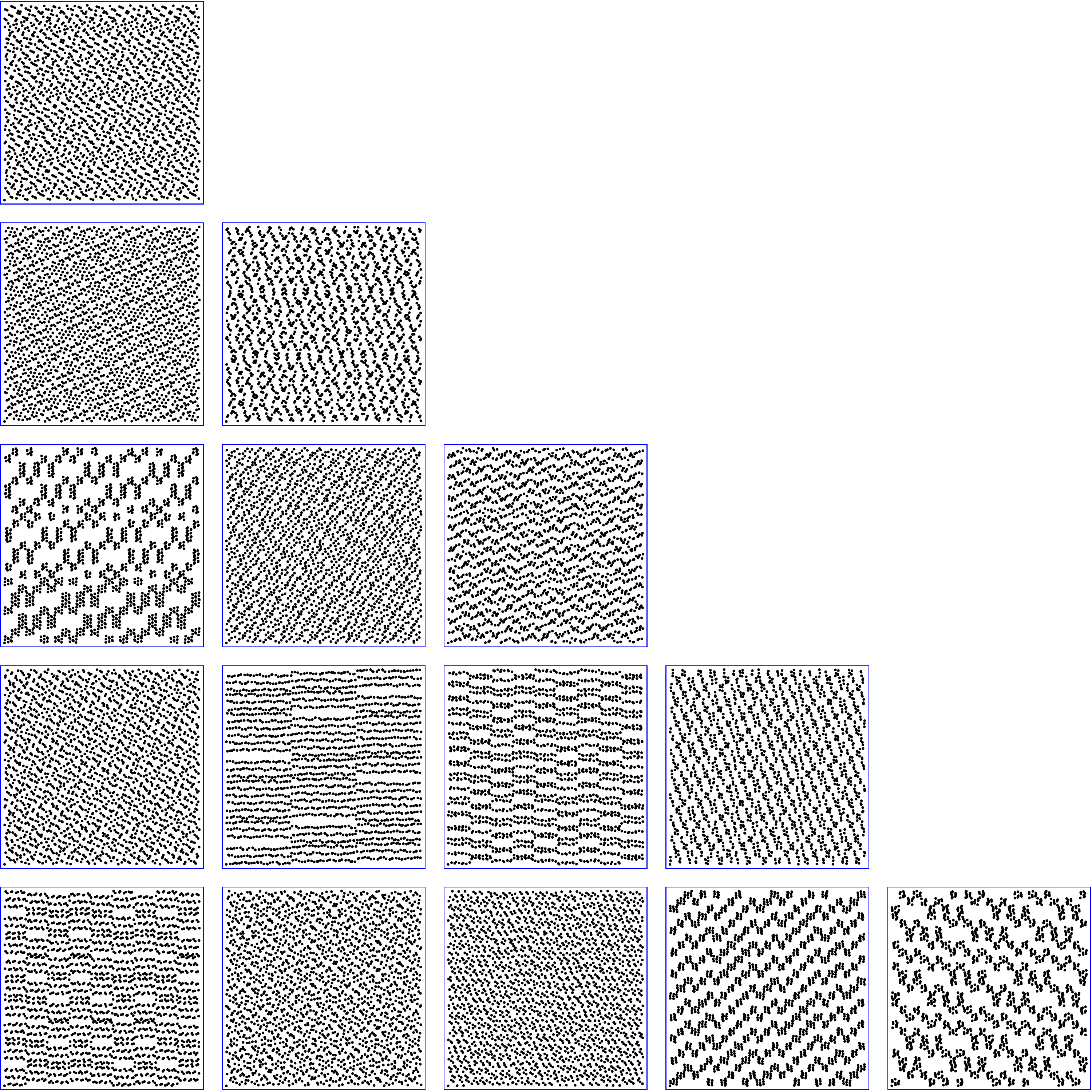}}
  \subfigure[]{\includegraphics[width=0.32\textwidth]{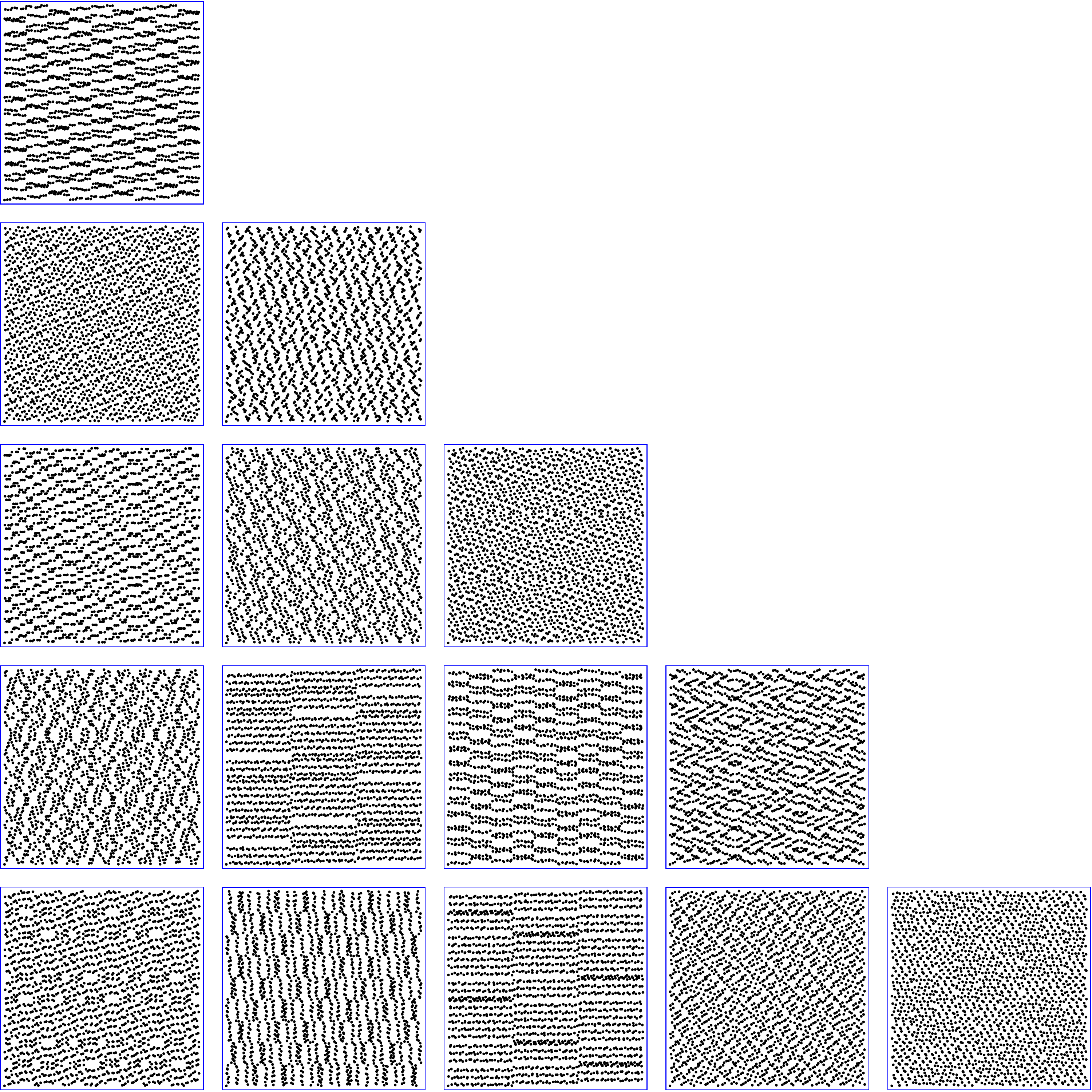}}
  \caption{Two-dimensional projections of $3^7$ points in dimension 6 following 
    \texttt{"weak 1 net ti 0 1 2 3 4 5"} profiles for \texttt{i} $\in \{0,1,2,3,4\}$ (from
    $(a)$ to $(e)$). Note that dimension indices follow the ones
    depicted in Fig.\ref{Fig:Projective}.}
  \label{Fig:Projective2}
\end{figure}

  \begin{figure}
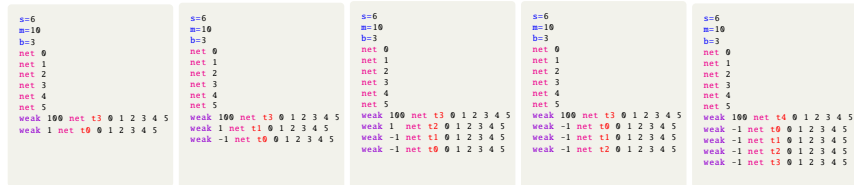

    \begin{center}
      \begin{minipage}{2.2cm}
        \begin{tcolorbox}[colback=backcolour,boxsep=-3mm,colframe=backcolour,arc=0mm]
          \begin{lstlisting}[style=mystyle3]
s=6
m=10
b=3
net 0
net 1
net 2
net 3
net 4
net 5
weak 100 net t3 0 1 2 3 4 5
weak 1 net t0 0 1 2 3 4 5




          \end{lstlisting}
        \end{tcolorbox}
      \end{minipage}
      \begin{minipage}{2.2cm}
        \begin{tcolorbox}[colback=backcolour,boxsep=-3mm,colframe=backcolour,arc=0mm]
          \begin{lstlisting}[style=mystyle3]
s=6
m=10
b=3
net 0
net 1
net 2
net 3
net 4
net 5
weak 100 net t3 0 1 2 3 4 5
weak 1 net t1 0 1 2 3 4 5
weak -1 net t0 0 1 2 3 4 5



          \end{lstlisting}
        \end{tcolorbox}
      \end{minipage}
      \begin{minipage}{2.2cm}
        \begin{tcolorbox}[colback=backcolour,boxsep=-3mm,colframe=backcolour,arc=0mm]
          \begin{lstlisting}[style=mystyle3]
s=6
m=10
b=3
net 0
net 1
net 2
net 3
net 4
net 5
weak 100 net t3 0 1 2 3 4 5
weak 1  net t2 0 1 2 3 4 5
weak -1 net t1 0 1 2 3 4 5
weak -1 net t0 0 1 2 3 4 5


          \end{lstlisting}
        \end{tcolorbox}
      \end{minipage}
      \begin{minipage}{2.2cm}
        \begin{tcolorbox}[colback=backcolour,boxsep=-3mm,colframe=backcolour,arc=0mm]
          \begin{lstlisting}[style=mystyle3]
s=6
m=10
b=3
net 0
net 1
net 2
net 3
net 4
net 5
weak 100 net t3 0 1 2 3 4 5
weak -1 net t0 0 1 2 3 4 5
weak -1 net t1 0 1 2 3 4 5
weak -1 net t2 0 1 2 3 4 5


          \end{lstlisting}
        \end{tcolorbox}
      \end{minipage}
      \begin{minipage}{2.2cm}
        \begin{tcolorbox}[colback=backcolour,boxsep=-3mm,colframe=backcolour,arc=0mm]
          \begin{lstlisting}[style=mystyle3]
s=6
m=10
b=3
net 0
net 1
net 2
net 3
net 4
net 5
weak 100 net t4 0 1 2 3 4 5
weak -1 net t0 0 1 2 3 4 5
weak -1 net t1 0 1 2 3 4 5
weak -1 net t2 0 1 2 3 4 5
weak -1 net t3 0 1 2 3 4 5
          \end{lstlisting}
        \end{tcolorbox}
      \end{minipage}
    \end{center}
    \caption{MatBuilder profiles exploring negatively weighted
        $t$ constraints: While enforcing the progressive net to be as
        $t=3$ as possible (and $t=4$ for the last one), we
      progressively invalidate some smaller $t$-parameter options.}
    \label{fig:tonlyprofiles}
  \end{figure}
  
Weighted weak constraints also enable us to negatively weigh a net constraint. While this seems to have little practical purpose, it allows one to explore the range of possible net configurations.
We revisit the example of sequences in 6 dimensions in base 3 where the smallest feasible $t$-parameter is $3$,
with the  profiles given in Fig.~\ref{fig:tonlyprofiles}.
One may understand such profiles as for instance ``What is the best
$t=3$ point set that is neither $t=0$, $1$, nor $2$?''. In
Figs.~\ref{Fig:Perf-t_only} and~\ref{Fig:NegativeNets}, we observe
that block artifacts on the projections have a direct impact on the
generalized $l_2$-discrepancy (with relatively similar timings for
the matrix construction). However,
Fig.~\ref{Fig:NegativeNets} shows that the $(3,6)$-sequence property
does not differentiate between points of high or low quality.

\begin{figure}[!h]
  \subfigure[]{\includegraphics[width=0.5\textwidth]{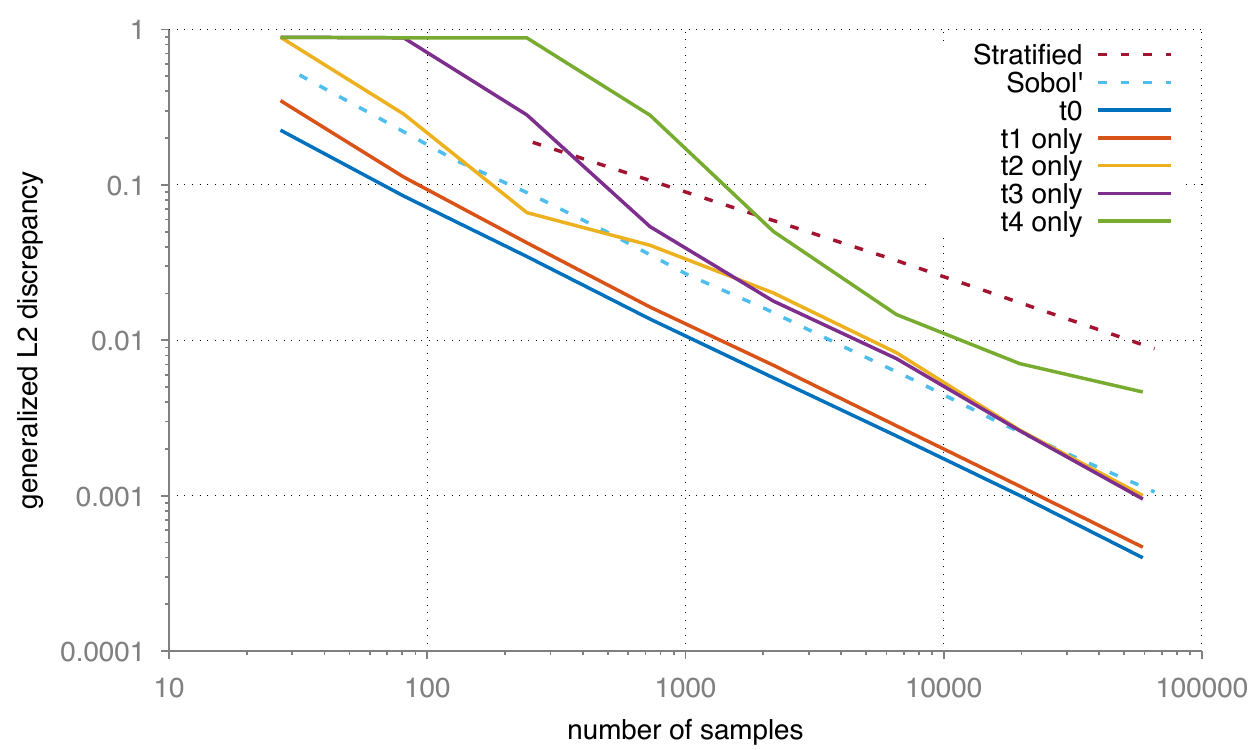}}
  \subfigure[]{\includegraphics[width=0.5\textwidth]{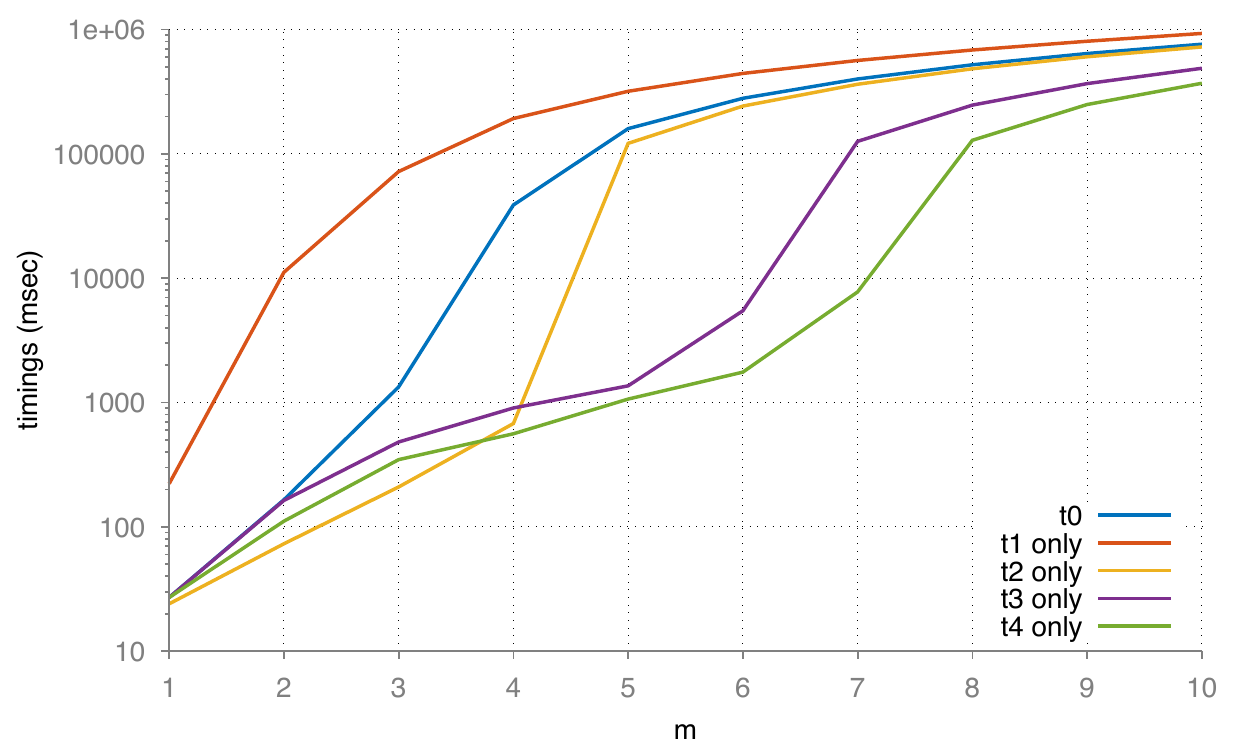}}
  \caption{Performance evaluation of the profiles given in Fig.~\ref{fig:tonlyprofiles}: $(a)$ Quality
    evaluation in terms of generalized $l_2$-discrepancy.  $(b)$ Timings of the
    solver as a function of the matrix size $m$.}
  \label{Fig:Perf-t_only}
\end{figure}

\begin{figure}[!h]\centering
  \subfigure[]{\includegraphics[width=0.32\textwidth]{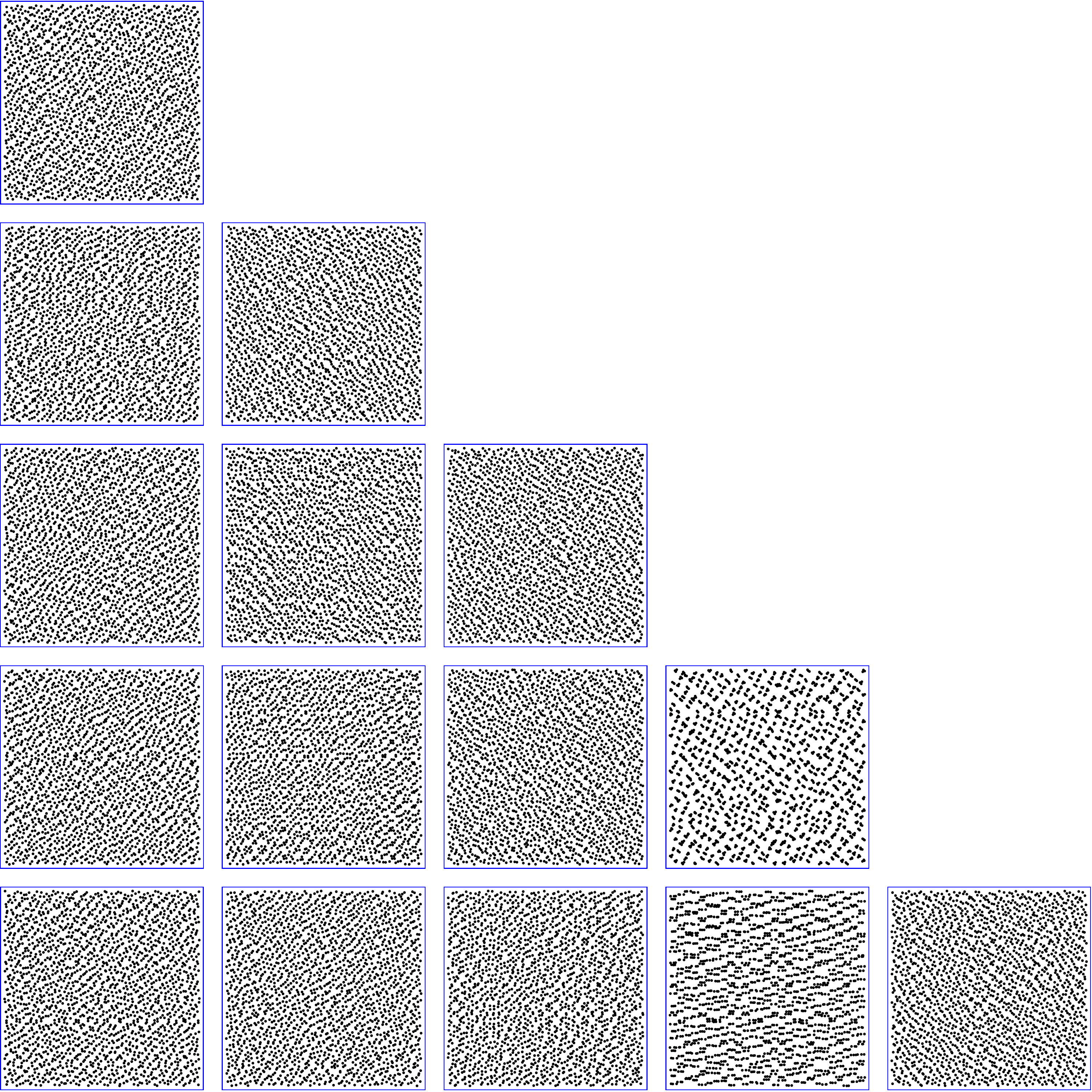}}
  \subfigure[]{\includegraphics[width=0.32\textwidth]{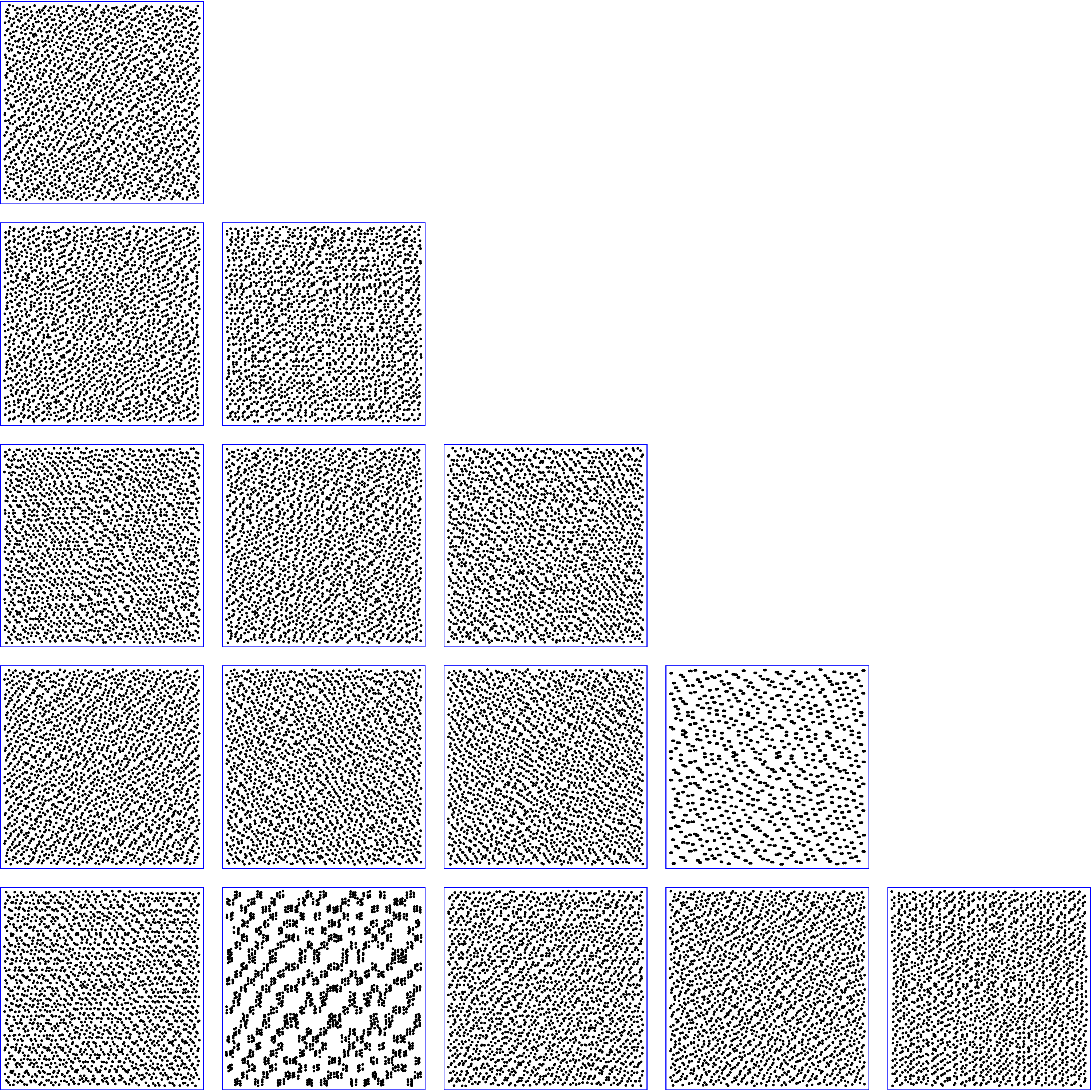}}
  \subfigure[]{\includegraphics[width=0.32\textwidth]{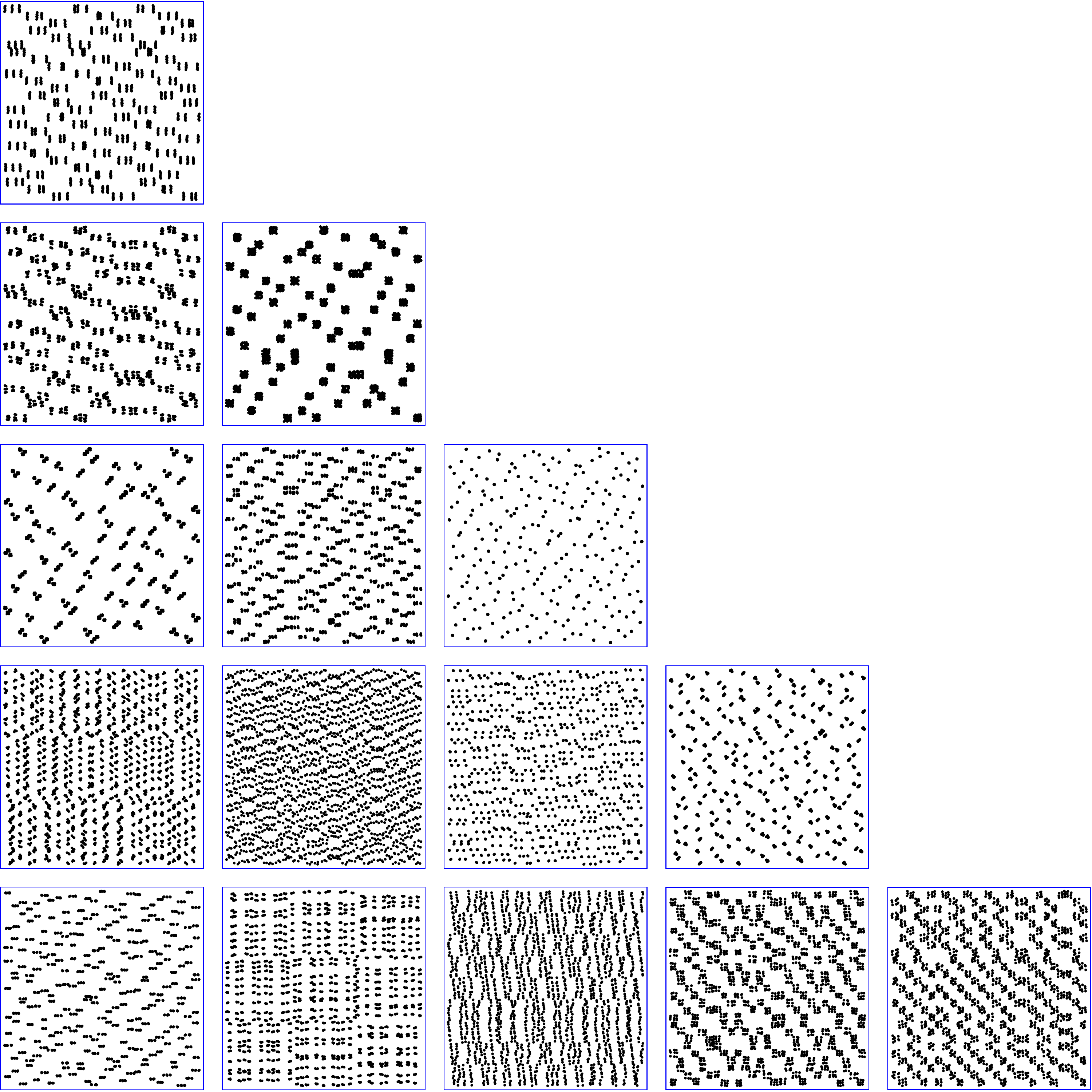}}\\
  \subfigure[]{\includegraphics[width=0.32\textwidth]{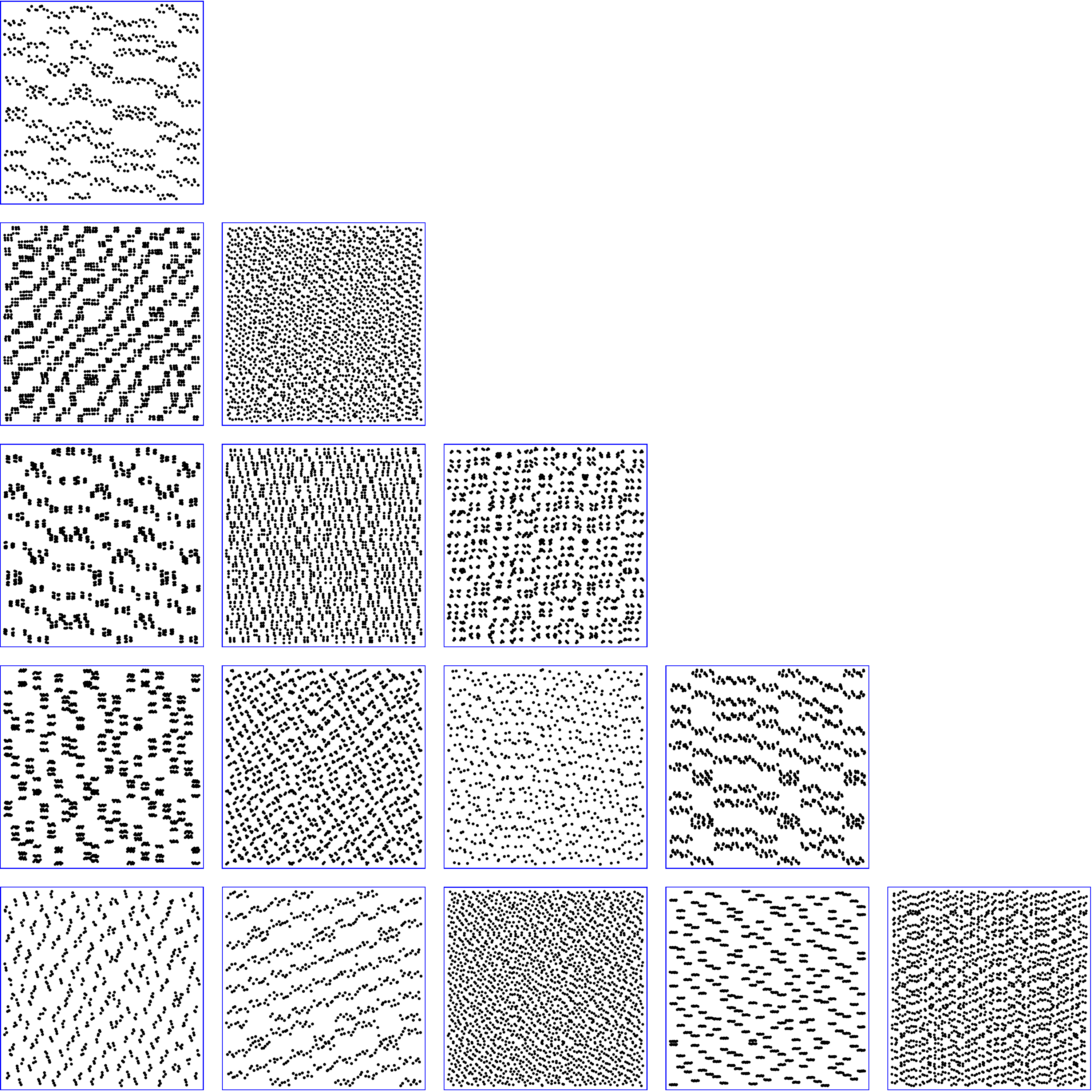}}
  \subfigure[]{\includegraphics[width=0.32\textwidth]{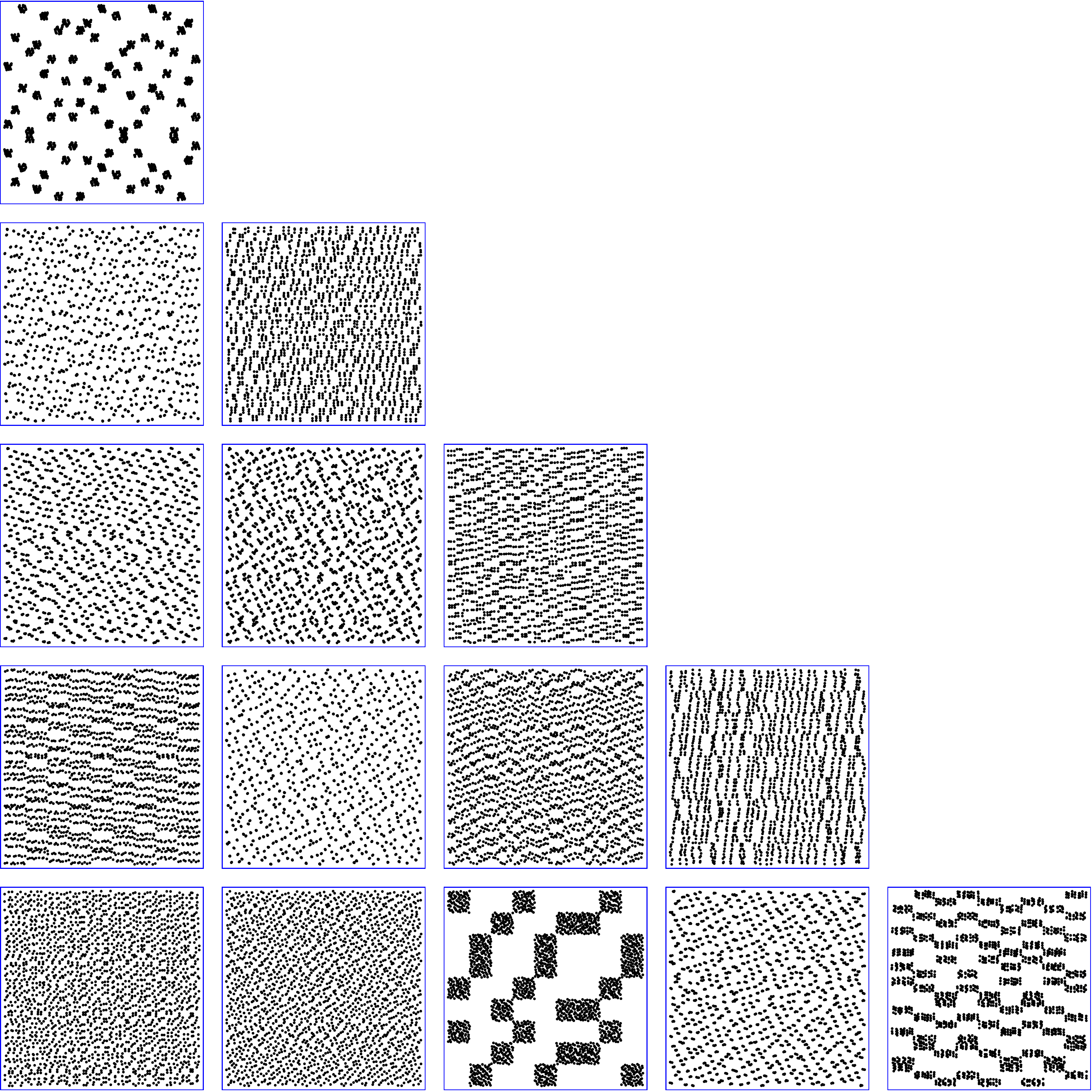}}
  \caption{Two-dimensional projections of $3^7$ points in dimension 6 resulting from 
    \texttt{"weak 1 net ti  0 1 2 3 4 5"} profiles for \texttt{i}
    $\in \{0,1,2,3,4\}$ with negatively weighted \texttt{"weak -1 net tj  0 1 2 3 4 5"}
    constraints for \texttt{j} $<$ \texttt{i} (from
    $(a)$ to $(e)$).}
  \label{Fig:NegativeNets}
\end{figure}

\section{Conclusion}

We extend the MatBuilder software to quality parameters $t > 0$.
This enables us to exemplify that the $(t,m,s)$-net property alone fails to characterize optimal quality both across projections as well as across all dimensions before reaching asymptotic behavior.
Using MatBuilder, we are confident that by exploring constraints on not necessarily disjoint subsets of dimensions, partially satisfying constraints,
and higher bases,
generator matrices can be found that outperform the classic constructions in practice.
Generalizing our approach to prime power Galois fields is a promising avenue of future research.

\bibliographystyle{spmpsci}
\bibliography{bibliography}

\end{document}